\documentclass[11pt]{article}

\usepackage[T1]{fontenc}
\usepackage[margin=1in]{geometry}
\usepackage{microtype}
\usepackage{amsmath,amssymb,amsthm,mathtools}
\usepackage{tikz}
\usetikzlibrary{arrows.meta,positioning}
\usepackage{enumitem}
\usepackage{cite}
\usepackage{authblk}
\usepackage[hidelinks]{hyperref}
\usepackage[capitalise,nameinlink,noabbrev]{cleveref}

\allowdisplaybreaks[2]
\setlist[itemize]{leftmargin=*,topsep=3pt,itemsep=2pt}
\setlist[enumerate]{leftmargin=*,topsep=3pt,itemsep=2pt}

\newtheorem{theorem}{Theorem}[section]
\newtheorem{lemma}[theorem]{Lemma}
\newtheorem{proposition}[theorem]{Proposition}
\newtheorem{corollary}[theorem]{Corollary}

\newtheorem*{informaltheorem}{Informal Theorem}
\theoremstyle{definition}
\newtheorem{definition}[theorem]{Definition}
\theoremstyle{remark}
\newtheorem{remark}[theorem]{Remark}

\newcommand{\C}{\mathbb C}
\newcommand{\Q}{\mathbb Q}
\newcommand{\bits}{\{0,1\}}
\newcommand{\cF}{\mathcal F}
\newcommand{\cG}{\mathcal G}
\newcommand{\cH}{\mathcal H}
\newcommand{\cP}{\mathcal P}
\newcommand{\cR}{\mathcal R}
\newcommand{\MG}{\operatorname{MG}}
\newcommand{\Hol}{\operatorname{Holant}}
\newcommand{\PerfMatch}{\operatorname{PerfMatch}}
\newcommand{\Pf}{\operatorname{Pf}}
\newcommand{\supp}{\operatorname{supp}}
\newcommand{\arity}{\operatorname{arity}}
\newcommand{\barity}{\operatorname{barity}}
\newcommand{\trdeg}{\operatorname{trdeg}}
\newcommand{\dist}{\operatorname{dist}}
\newcommand{\Mat}{\operatorname{Mat}}
\newcommand{\bw}{\operatorname{bw}}
\newcommand{\ind}{\operatorname{ind}}
\newcommand{\one}{\mathbf 1}
\newcommand{\trans}{\mathsf T}
\newcommand{\Contr}{\operatorname{Contr}}
\newcommand{\row}{\operatorname{row}}
\newcommand{\Ann}{\operatorname{Ann}}
\newcommand{\Span}{\operatorname{span}}
\newcommand{\rank}{\operatorname{rank}}

\newcommand{\sgn}{\operatorname{sgn}}

\title{When Matchgate Base Collapse Fails:\\
A Qutrit Trichotomy and Unbounded Exact Width}
\author[1]{Chenghua Liu}
\author[2]{Boning Meng}

\affil[1]{Institute of Software, Chinese Academy of Sciences, Beijing, China}
\affil[2]{University of Regensburg, Regensburg, Germany}
\affil[ ]{\texttt{liuch.russell@gmail.com}, \texttt{mengboning2013@gmail.com}}
\date{}

\begin{document}

\maketitle

\begin{abstract}
Holographic algorithms solve planar counting problems by encoding each value
of a finite domain into several Boolean matchgate wires.  Base collapse asks
whether every exact representation can be compressed to a number of wires per
edge bounded only by the domain size.  Chen (STOC 2016) and Xia (STOC 2016)
established broad positive collapse theorems under full-rank and related
structural hypotheses.  Together, these works left open whether a domain-only
collapse bound survives when several locally deficient signatures must share
one encoding.  We resolve this open problem negatively, already on a
three-state (qutrit) domain.  For every $a\ge4$, we construct a five-label
integer-valued qutrit language of maximum arity $a$ with an exact rational
matchgate presentation, yet every exactly equivalent label-preserving
presentation requires common width $\Omega(2^{a/2}/a)$, even if its domain,
base, tensors, and complex weights may all change.  The construction avoids
the usual local degeneracies, so the obstruction is genuinely simultaneous.
Algebraic geometry then drives an exhaustive trichotomy explaining the
boundary of collapse: the gadget closure either separates into rays, is
Gaussian-mobile and collapses to width at most three, or is confined to a
rigid plane-plus-ray flag containing our unbounded family.  Pure-spinor
geometry forces compression in the mobile case, while dimension bounds for
matchgate varieties and Zariski avoidance place exact integer data outside
every low-width algebraic image.  Thus one geometric framework explains both
why collapse occurs and why it can fail without bound.
\end{abstract}


\section*{Acknowledgments}
OpenAI Codex's principal substantive contribution to the mathematical
development was to help identify, introduce, and formulate the
algebraic-geometric viewpoint used in the proofs.  In particular, it suggested
treating the Zariski closures of low-width star-table images as
$\Q$-defined algebraic varieties, using finite Zariski avoidance to select an
integer obstruction table, and applying pure-spinor incidence and
common-annihilator geometry to obtain low-rank matchgate covers.  The authors
checked, revised, and integrated these suggestions into the final arguments.

\section{Introduction}
\label{sec:introduction}

Holographic algorithms turn finite-domain planar counting problems into
Boolean networks of planar matchgates, where the global sum is evaluated by
Pfaffians and weighted perfect matchings.  Matchgates first arose in
Valiant's study of classically simulable quantum circuits and later became
the computational engine of holographic algorithms
\cite{Valiant2002,Valiant2006,Valiant2008}.  In broad Boolean counting-CSP
settings they account precisely for the tractability created by planarity
\cite{CaiLuXia2017}.  Understanding their representations is therefore part
of understanding why this form of exact cancellation works at all.

The change of representation is mediated by a \emph{base}.  A width-$t$
base assigns every original domain value a vector in a $2^t$-dimensional
Boolean space: operationally, each domain edge is replaced by $t$ Boolean
matchgate wires.  Crucially, the same encoding must work at every occurrence
of every label.  Extra wires may therefore be mere coordinate redundancy,
or they may store compatibility information that is visible only when many
labelled tensors are represented simultaneously.

The central collapse question is whether every matchgate-realizable labelled
language on a $q$-element domain admits a common presentation whose width is
bounded by a function of $q$ alone.  We answer it negatively: already for
$q=3$, minimum exact common width is unbounded.  This resolves the
uniform-width aspect of the multi-signature problem left open after Cai and
Fu's full-rank-generator theorem~\cite{CaiFu2014} and Chen's all-domain
full-rank theorem~\cite{Chen2016}.  Xia supplied
further collapse certificates but explicitly isolated this gap: base
directions unused by one signature may still be needed by another, and the
attainable common collapse size was unknown
\cite[end of Section~5]{Xia2016}.

If a domain-only bound existed, arbitrarily large ambient matchgate spaces
would add no exact expressive power at fixed domain size, and the search for
holographic algorithms could be restricted to bounded-width bases.  Our
result shows instead that simultaneous compatibility can carry irreducibly
large information.  The question is therefore about the intrinsic size of a
shared representation, not merely about shortening one displayed base
matrix.

Matchgate signatures also have two complementary descriptions.  They are
boundary tensors of planar perfect-matching gadgets, and, after Boolean
indices are read as occupation modes, they are Pfaffian or
fermionic-Gaussian tensors, subject to the usual parity conventions
\cite{Bravyi2009,JozsaMiyake2008}.  Base collapse can accordingly be viewed
as a structured common-mode compression problem for an exact Gaussian
tensor language.  Ordinary linear algebra alone does not answer it: a
qutrit fits inside two Boolean bits, but an entire labelled qutrit language
need not become matchgate under one small common encoding.

\subsection{The collapse question and our answer}

We measure representation size by \emph{minimum exact common width}: the
fewest Boolean wires per domain edge among all matchgate presentations that
preserve the finite label families, their arities and port orders, and the
exact value of every planar labelled instance.  A competing presentation may
change the finite domain, the common base, every coordinate tensor, and all
complex weights.  It may not replace a label by a gadget or merely preserve
polynomial-time computability.  The precise notions appear in
\cref{sec:equivalence-definition}.

\begin{informaltheorem}[Unbounded exact width on a qutrit domain]
For every $k\ge2$ there is a five-label integer-valued qutrit language
$\cF_k\mid\{B_k\}$ whose unique right label $B_k$ has arity $k+2$.  It has
an exact rational matchgate presentation, but its minimum exact common width
is
\[
  \Omega\!\left(\frac{2^{k/2}}{k}\right).
\]
The displayed presentation can simultaneously be chosen with a full-row-rank
base, support-essential left and right families, rank exactly two at every
port of every nonunary primitive, and a connected two-center gadget whose
boundary tensor is nonfactorable across the two centers.
\end{informaltheorem}

There are only five labels: the three unary qutrit pins, binary disequality
on states $0,1$, and one right label $B_k$ of arity $k+2$.  Thus the
obstruction is not a proliferation of unrelated primitives.  Moreover every
left label is symmetric in its ports, so each singleton left subproblem is
compatible with the type of symmetry used in Xia's singleton
collapse~\cite[Theorem~5.3]{Xia2016}.  The
large width is genuinely simultaneous: it is the requirement that the whole
labelled language share one exact representation that becomes costly.

The result rules out any width bound depending only on the domain size in
this unrestricted model.  It is stronger than a lower bound against
compressing the displayed base, because it survives arbitrary exact
re-presentation.  It does not assert computational hardness: the high-arity
right label grows with $k$, and our equivalence notion is stricter than gadget
or polynomial-time equivalence.  The full statement and quantitative
inequality are given in \cref{thm:main}, and the domain-only consequence is
\cref{cor:no-domain-bound}.

\subsection{Where can non-collapse occur?}

The negative result leaves a structural question: why can width hide in a
three-state language at all?  Consider actual connected planar gadgets whose
dangling edges are incident with left vertices; their internal vertices may
use labels from either side.  At an exposed port, the \emph{mode support} is
the span of all qutrit vectors obtainable there as the other ports vary.  A
port is deficient when this span misses at least one qutrit direction, so a
nonzero support is either a one-dimensional \emph{ray} or a two-dimensional
\emph{plane}.  Bi-essentiality says that the mode supports of the primitive
left labels, and separately those of the induced right labels, collectively
see all three qutrit directions.

After the base is applied, matchgate parity singles out two endpoint rays
$P^+,P^-$ inside every realized plane $P$; these endpoints need not
themselves occur as gadget supports.  The exceptional \emph{monomial flag}
has all realized planes equal to one plane $P$, contains a transverse realized
ray $L$, and confines every realized ray to $P^+,P^-$, or $L$.

\begin{informaltheorem}[Support-geometric trichotomy]
For a full-row-rank, bi-essential, port-deficient qutrit presentation,
exactly one of the following occurs.
\begin{enumerate}[label=\textup{(\Roman*)}]
  \item \emph{Ray-separable.}  No plane is realized, and every nonzero
  positive-arity all-left boundary tensor is a pure tensor, that is, it
  factors algebraically into unary vectors.
  \item \emph{Gaussian-mobile.}  A plane is realized, but the supports are
  not confined to one flag.  The whole labelled language has an exact common
  presentation of width at most three; if two distinct planes occur, the
  width is at most two.
  \item \emph{One-plane flag.}  All realized planes equal $P$ and all
  realized rays lie among $P^+,P^-,L$.  Every nonzero positive-arity
  all-left boundary tensor is a Boolean matchgate core on $P$ times pure-ray
  factors; a nullary boundary tensor is a scalar.
\end{enumerate}
The third branch contains a family whose minimum exact widths are unbounded.
\end{informaltheorem}

This is the significance of the trichotomy.  Once a plane-bearing support
configuration escapes every monomial flag---through a second plane or
additional transverse rays---Gaussian geometry produces one bounded-width
encoding shared by every label.  The flag still does not force large width for each individual
language, and the ray-only branch is not yet a width classification.  The
branches are properties of a displayed presentation, whereas the collapse
obtained in branch~\textup{(II)} is a new presentation of the same labelled
language.  Our unbounded family is constructed inside branch~\textup{(III)}.
There the flag rigidifies the all-left closure without controlling the
residual right table.  The label $B_k$ can store a $2^k$-entry residual
table, which we choose outside every relevant low-width image.
The precise statement is \cref{thm:structural-trichotomy}.

\subsection{Proof architecture and technical ideas}

The negative construction and the positive boundary use five main ideas.
\begin{enumerate}
  \item \emph{Composable-support geometry for base collapse.}
  Mode supports cannot increase when planar gadgets are composed.  Primitive
  port deficiency therefore constrains an entire gadget closure by a small
  incidence geometry of realized rays and planes.

  \item \emph{Pure-spinor hulls as common-cover certificates.}
  The base sends realized rays to pure matchgate rays and realized planes to
  rank-two Gaussian row spaces.  Two distinct planes, or two distinct pure
  rays outside one realized plane, have a low-codimension common annihilator.
  This produces a rank-four or rank-eight matchgate cover---one Gaussian row
  space containing the whole base row space---and hence collapse.

  \item \emph{All-arity flag-normal-form decoding.}
  In the one-plane branch, rank-one ports can be stripped off and the
  remaining plane-supported tensor can be decoded through one fixed pair of
  parity endpoints.  This yields a Boolean matchgate core times pure-ray
  factors uniformly at every arity.

  \item \emph{Universal block-Pfaffian interpolation.}
  Any everywhere-nonzero Boolean table can be interpolated on two repetition
  codewords per block by a single exact matchgate.  A multiplicative
  M\"obius decomposition turns the table entries into independent Pfaffian
  controls and supplies the large-width presentation.

  \item \emph{A representation-independent star-table dimension
  obstruction.}
  The variety of $s$-wire matchgate signatures has dimension at most
  $1+\binom{s}{2}$.  Closed labelled stars extract the same $2^k$-entry
  table from every exactly equivalent presentation, regardless of its new
  domain or base.  Zariski avoidance then chooses an integer table outside
  every lower-width star-table image and gives the quantitative bound.
\end{enumerate}

The fourth idea realizes arbitrary information; the fifth proves that this
information cannot be hidden in fewer common modes.  The first three explain
why the construction must be placed in the fixed flag and why nearby support
geometries collapse.  The two proof lanes and their locations in the paper
are summarized below.

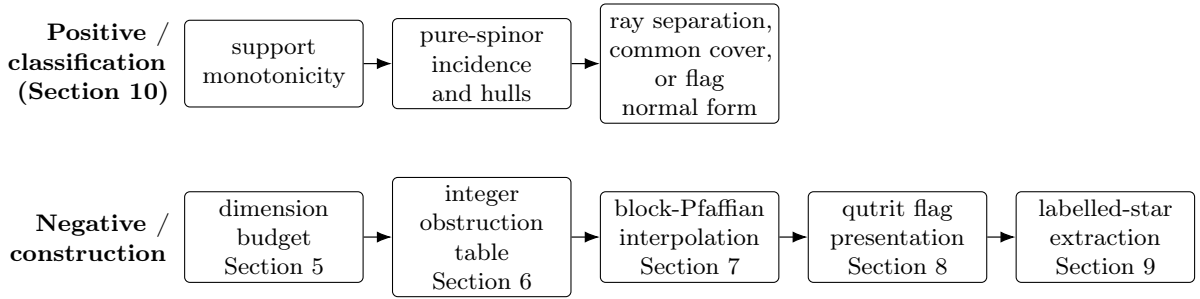
\begin{figure}[htbp]
  \centering
  \begin{tikzpicture}[
      >={Latex[length=2mm]},
      lanehead/.style={font=\bfseries\footnotesize,align=right,text width=2.35cm},
      roadstep/.style={draw,rounded corners=2pt,align=center,font=\footnotesize,
        text width=2.15cm,minimum height=1.15cm,inner sep=3pt},
      arr/.style={->,line width=.45pt}
    ]
    \node[lanehead] at (-1.45,1.15)
      {Positive /\\classification\\ (\cref{sec:positive-boundary})};
    \node[roadstep] (p1) at (1.15,1.15)
      {support monotonicity};
    \node[roadstep] (p2) at (3.90,1.15)
      {pure-spinor incidence\\and hulls};
    \node[roadstep] (p3) at (6.65,1.15)
      {ray separation,\\common cover,\\or flag normal form};
    \draw[arr] (p1)--(p2);
    \draw[arr] (p2)--(p3);

    \node[lanehead] at (-1.45,-1.15) {Negative /\\construction};
    \node[roadstep] (n1) at (1.15,-1.15)
      {dimension\\budget\\\cref{sec:matchgate-varieties}};
    \node[roadstep] (n2) at (3.90,-1.15)
      {integer obstruction\\table\\\cref{sec:obstruction-tables}};
    \node[roadstep] (n3) at (6.65,-1.15)
      {block-Pfaffian\\interpolation\\\cref{sec:pfaffian-interpolation}};
    \node[roadstep] (n4) at (9.40,-1.15)
      {qutrit flag\\presentation\\\cref{sec:domain-three-construction}};
    \node[roadstep] (n5) at (12.15,-1.15)
      {labelled-star\\extraction\\\cref{sec:arbitrary-equivalence}};
    \draw[arr] (n1)--(n2);
    \draw[arr] (n2)--(n3);
    \draw[arr] (n3)--(n4);
    \draw[arr] (n4)--(n5);
  \end{tikzpicture}
  \caption{Conceptual roadmap.  The upper lane classifies support geometry;
    the lower lane constructs information and proves that no exactly
    equivalent presentation can compress it below the stated width.}
  \label{fig:conceptual-roadmap}
\end{figure}

\subsection{Relation to previous work}

The collapse line begins with Cai and Lu's Boolean-domain results
\cite{CaiLu2008,CaiLu2009}.  Fu and Yang treated simultaneous realizability
over rank-two bases~\cite{FuYang2014}.  Later work studied blockwise-symmetric
signatures and holographic algorithms over general domains
\cite{FuYangYin2019,FuCai2023}.

The closest positive results delimit our regime.  The broad formulation of
Cai and Fu's domain-three collapse is used under their standing
full-rank-generator hypothesis; in the degenerate alternative their
conclusion is a direct computational simplification, rather than a
domain-only bound on minimum exact common width~\cite{CaiFu2014}.  Chen's
all-domain theorem likewise covers the regime of a full-rank original
signature~\cite{Chen2016}.  Xia supplied several further sufficient
certificates: a full-rank realizer, a bounded-rank common matchgate cover, or
a singleton symmetric left family~\cite{Xia2016}.  Our displayed
presentations lie outside these regimes: their bases have full row rank and
the two induced families are support-essential, but every port of every
nonunary qutrit primitive has rank exactly two, and the simultaneous left
family is not a singleton.

Xia also identified the precise compatibility gap that motivates our work.
His singleton proof may discard a part of the base unused by one signature,
but that part can remain necessary for another; after a two-signature example
outside his hypotheses, he wrote that no characterization of matchgate
covers was known and that the achievable collapse size was unknown
\cite[end of Section~5]{Xia2016}.  Under exact same-instance equivalence,
\cref{cor:no-domain-bound} shows that the attainable width cannot be bounded
as a function of the domain size alone.  This answers the uniform-width
aspect of Xia's question, not the characterization of matchgate covers or the
optimal collapse size of each individual language.

Our structural result also identifies the boundary within this port-deficient
qutrit regime.  Support mobility yields an exact common presentation of width
at most three---and at most two when two distinct planes occur---whereas the
one-plane flag contains languages of unbounded minimum exact common width.

Nontrivial negative results for holographic algorithms did exist before this
work.  Cai and Choudhary used matchgate identities to rule out holographic
templates for natural problem generalizations~\cite{CaiChoudhary2007}, and
Fu, Yang, and Yin proved nonrealizability of higher-domain equality under
holographic transformations~\cite{FuYangYin2019}.  Those results exclude a
specified realization.  Here every language is already matchgate-realizable;
the question is how small any exact common realization can be.  To our
knowledge, \cref{thm:main} gives the first unbounded lower bound on minimum
common matchgate-base width under exact labelled equivalence, even over a
fixed three-element domain.

Algebraic geometry is likewise not new to matchgate theory.  Matchgate
identities define algebraic varieties, basis manifolds have been used for
simultaneous realizability, and spinor/Pfaffian geometry has supplied a
conceptual language for holographic algorithms
\cite{CaiChoudharyLu2009,CaiLu2011,LandsbergMortonNorine2013,Bravyi2009}.
Bradley also used dimension counts to prove lower bounds on the graph size
needed to realize standard signatures~\cite{Bradley2009}; that resource is
the size of an individual matchgate, rather than the width of a common base.
The new role here is to turn that geometry into a lower bound robust under
arbitrary exact re-presentation: pure-spinor incidence controls the collapse
side, while star-table varieties and Zariski avoidance control the
non-collapse side.  In this sense algebraic geometry is a central method of
the paper, though not a first use of the subject in holographic algorithms.

\subsection{Organization}

\Cref{sec:statement-preliminaries} gives the minimum terminology needed to
read the formal statements in \cref{sec:main-results}.
\Cref{sec:technical-preliminaries} then develops the proof machinery.
Sections \ref{sec:matchgate-varieties}--\ref{sec:arbitrary-equivalence}
build the obstruction and prove the lower bound;
\cref{sec:positive-boundary} proves the trichotomy, and
\cref{sec:discussion} collects the open directions.

\section{Statement-Level Preliminaries}
\label{sec:statement-preliminaries}

This section isolates only the notions needed to parse the main results.
The complete coordinate conventions, ordered planar operations, and
proof-oriented definitions are developed in
\cref{sec:technical-preliminaries}.  A reader interested first in the
statements may therefore proceed directly from this section to
\cref{sec:main-results}.

\subsection{Signatures, planar instances, and local rank}
\label{sec:statement-signatures}

Let $D$ be a finite nonempty domain, take
$\mathbb N=\{1,2,\ldots\}$, and put $D_3=\{0,1,2\}$.  Write $e_a$ for the
coordinate vector of $a\in D$ and $\delta_a$ for the corresponding unary
pin.  A signature
of arity $a$ is a tensor $F\colon D^a\to\C$ whose ordered arguments are its
\emph{ports}.  All families are finite and labelled: two labels remain
distinct even if their coordinate tensors agree.  A bipartite language
$\cF\mid\cG$ places labels from $\cF$ and $\cG$ on the two sides of a
planar embedded bipartite graph.  Its value is the sum over all edge
assignments of the product of the incident tensor entries.  Exact labelled
equivalence below preserves this value on the same embedded, ordered,
labelled graph.

At port $j$, the \emph{mode flattening} $\Mat_j(F)$ is the matrix with the
$j$th coordinate as row index and all other coordinates as column index.
Its image $S_j(F)$ is the \emph{mode support}, and its rank is the
\emph{mode rank}.  A port is \emph{deficient} when this rank is smaller
than $|D|$, and a signature is \emph{non-full-rank} when every port is
deficient.  Thus a qutrit port is deficient exactly when its rank is at
most two.  A family is \emph{support-essential} when the sum of all of its
mode supports is the full domain space (three-dimensional for a qutrit); a bipartite
language is \emph{bi-essential} when this holds on both sides.  A nonzero
tensor is \emph{completely decomposable}, or \emph{pure}, when it is a
tensor product of unary vectors.  Its \emph{Schmidt rank} across a partition
of the ports is the rank of the corresponding two-block flattening.  These
notions and the left/right variance convention are formalized in
\cref{sec:tensor-preliminaries,sec:holant-preliminaries}.

\subsection{Common matchgate presentations}
\label{sec:statement-presentations}

An exact $s$-wire matchgate signature is the boundary tensor of a finite
weighted planar perfect-matching gadget with $s$ ordered Boolean ports; the
set of such tensors is denoted $\MG_s$.  A width-$t$ base on $D$ is a
matrix
\[
  M\in\C^{|D|\times2^t}.
\]
It replaces every domain edge by one ordered block of $t$ Boolean wires.
For a left signature $F$ of arity $a$, the Boolean lift is
$FM^{\otimes a}$.  A right domain signature is written
$G=M^{\otimes b}H$, where $H$ is its Boolean preimage and $b$ is its block
arity.  A \emph{valid common matchgate presentation} $(\cF,\cH,M)$ requires
every $FM^{\otimes a}$ and every $H$ to be an exact matchgate in the
prescribed block order, with the same base $M$ used for all labels.  Its
induced domain language is $\cF\mid M\cH$.  The presentation is rational
when all displayed coordinates and realizing edge weights may be chosen in
$\Q$.  Full definitions and the ordered contraction convention appear in
\cref{sec:matchgate-preliminaries,sec:base-presentation}.

An $r$-input, $t$-output exact matchgate matrix is the matrix obtained by
viewing $r$ ordered Boolean ports of an exact matchgate as inputs and $t$ as
outputs, with the fixed planar output reversal.  A \emph{rank-$2^r$ common
matchgate cover} of $M$ is such a matrix $P$ satisfying
\[
  \row(M)\subseteq\row(P).
\]
A common cover can be decoded to a width-$r$ presentation.  The matrix,
cover, and decoding conventions are given in
\cref{sec:matchgate-matrix-preliminaries}.

\subsection{Exact labelled equivalence and minimum width}
\label{sec:equivalence-definition}

\begin{definition}[Exact labelled equivalence]
\label[definition]{def:exact-labelled-equivalence}
Let $\cF\mid\cG$ be a finite-domain planar language over $D$, and let
$\cP\mid\cR$ be one over an arbitrary finite domain $E$.  They are
\emph{exactly equivalent} if there are arity-preserving bijections
between the two left label sets and between the two right label sets such
that every planar embedded $\cF\mid\cG$ instance has exactly the same
complex value after replacing each label by its counterpart on the same
embedded graph.  Incidences, marked first ports, and linearized cyclic port
orders are unchanged.
\end{definition}

This is the planar, explicitly port-ordered version of Xia's
\emph{result equivalence}, which likewise uses label bijections and equality
on the same instance graph~\cite[Section~2.3]{Xia2016}.  This notion does
not permit replacing a label by a gadget, changing an
arity, preserving values only up to a scalar, or merely retaining
polynomial-time computability.  It does permit an alternative presentation
to change the finite domain, base, coordinate tensors, and exact complex
weights.  Thus it is a language-level notion, whereas ranks, supports, and
flags associated with one displayed triple are presentation-relative.

\begin{definition}[Minimum exact common width]
\label[definition]{def:base-width}
For a finite-domain language $\cF\mid\cG$ admitting at least one valid
presentation, define $\bw(\cF\mid\cG)$ to be the minimum integer $r\ge0$
for which some valid common width-$r$ presentation over some finite
nonempty domain induces a labelled language exactly equivalent to
$\cF\mid\cG$.
\end{definition}

No rank condition is imposed on an alternative base or its domain
signatures.  At width zero the Boolean block alphabet is the singleton
$\bits^0=\{\epsilon\}$, so declared block arity still records the induced
domain arity.  The minimum ranges over all such alternative presentations,
not merely submatrices or coordinate changes of a displayed base.  A \emph{uniform}
width bound controls every language in a stated class by one function of
the specified parameters; it is stronger than a pointwise assertion about
each fixed finite language.

\subsection{Realized supports and monomial flags}
\label{sec:statement-supports}

Fix a qutrit presentation with $\rank M=3$.  A connected planar gadget is
\emph{all-left} when every dangling edge is incident with a left vertex.
At a distinguished dangling port of a nonzero all-left boundary tensor, its
mode support is a \emph{realized support}.  Let $\mathcal S_1$ and
$\mathcal S_2$ be the sets of realized rays and planes, respectively.
For every realized plane $P$, the transformed block support $PM$ is a
rank-two exact-matchgate row space.  Its even- and odd-parity lines pull back
to two distinguished rays of $P$, denoted $P^+$ and $P^-$.  A pair
$(P,L)$, where $P\in\mathcal S_2$ and $L\in\mathcal S_1$ is transverse to
$P$, is a \emph{monomial flag} when
\[
  \mathcal S_2=\{P\},
  \qquad
  \mathcal S_1\subseteq\{P^+,P^-,L\}.
\]

In the flag branch, a \emph{Boolean matchgate core on $P$ times pure-ray
factors} means that the plane-supported ports are encoded, through one
fixed basis of the two parity endpoints, by a Boolean matchgate tensor,
while every remaining port contributes a unary factor on one of the three
allowed rays.  The exact gadget, rerooting, and normal-form conventions are
given in \cref{sec:support-terminology}.

\section{Main Results}
\label{sec:main-results}

\subsection{Unbounded exact width}

For $a,b\in D_3$, let $E_{ab}$ denote the qutrit matrix unit, and put
$X_{01}:=E_{01}+E_{10}$; thus $X_{01}$ is binary disequality on states
$0,1$ and vanishes whenever one input is $2$.

\begin{theorem}[Unbounded exact width with a flag presentation]
\label{thm:main}
For every integer $k\ge2$, there is a finite integer-valued domain-three
language
\[
  \cF_k=\{\delta_0,\delta_1,\delta_2,X_{01}\},
  \qquad \cG_k=\{B_k\}.
\]
The right label $B_k$ has arity $k+2$, which is therefore the maximum arity
of the language.  Put
\begin{equation}
  \mu_k:=\bw(\cF_k\mid\cG_k).
  \label{eq:muk-definition}
\end{equation}
Then the following properties hold.
\begin{enumerate}[label=\textup{(\roman*)}]
  \item It has an exact rational matchgate presentation with a full-row-rank
  $3$-by-$2^{2^k+3}$ base, of width $2^k+3$.
  \item Every port of $X_{01}$ and $B_k$ has flattening rank exactly two,
  whereas the pins have rank one.
  \item The language is support-essential on both sides.  The single right
  tensor $B_k$ has two link-port supports $\Span(e_0,e_1)$ and $k$
  hard-port supports $\Span(e_0,e_2)$---the link ports are its first two
  ports and the hard ports its remaining $k$---so it is support-essential
  by itself.
  \item Two copies of $B_k$ occur in a connected planar gadget whose
  boundary tensor has Schmidt rank exactly two across the hard ports of the
  first copy versus those of the second copy.  Thus the construction is not
  a product of isolated stars.
  \item The minimum exact width satisfies
  \begin{equation}
    2^k\le
    2\left(1+\binom{\mu_k}{2}\right)
    +1+\binom{k\mu_k}{2}.
    \label{eq:main-exact-lower-bound}
  \end{equation}
  Consequently,
  \begin{equation}
    \mu_k\ge
    \left\lceil\sqrt{\frac{2^k-3}{1+k^2/2}}\right\rceil
    =\Omega\!\left(\frac{2^{k/2}}{k}\right).
    \label{eq:main-readable-lower-bound}
  \end{equation}
\end{enumerate}
The presentation in part~\textup{(i)} may be chosen to lie in
alternative~\textup{(III)} of \cref{thm:structural-trichotomy}.
\end{theorem}

The quantification in \cref{thm:main} is unrestricted: a competing
presentation may change the finite domain, the common base, and every
coordinate tensor, provided that the original labels, arities, ordered
ports, and all planar instance values are preserved exactly.  Thus the
lower bound is not merely against compressions or invertible coordinate
changes of the displayed base.  The construction gives
$\mu_k\le2^k+3$; no matching upper bound is claimed.

\begin{corollary}
\label[corollary]{cor:no-domain-bound}
Under exact labelled equivalence, there is no function
$h\colon\mathbb N\to\mathbb N$ such that, for every finite nonempty domain
$D$ and every valid presentation $(\cF,\cH,M)$ on $D$ for which every
signature in the induced domain families $\cF\cup M\cH$ is integer-valued,
has positive arity, and is non-full-rank, one has
$\bw(\cF\mid M\cH)\le h(|D|)$.  This remains false even when the displayed
presentation is over $\Q$; equivalent presentations in the definition of
$\bw$ may still use arbitrary exact complex weights.
\end{corollary}

\subsection{Support-geometric trichotomy}

The lower bound is complemented by an exhaustive classification of the
qutrit support geometries relevant to base collapse.

\begin{theorem}[Support-geometry trichotomy and an unbounded flag family]
\label{thm:structural-trichotomy}
Let $(\cF,\cH,M)$ be a valid common matchgate presentation on a qutrit
domain with $\rank M=3$, and let $\mathcal B=M\cH$ be the induced right
family.  Assume that $\cF\mid\mathcal B$ is bi-essential and every port
flattening of every domain tensor in $\cF\cup\mathcal B$ has rank at most
two.  Let $\mathcal S_1,\mathcal S_2$ be the supports realized at
distinguished ports of actual connected all-left planar gadgets.  Exactly
one of the following alternatives holds.
\begin{enumerate}[label=\textup{(\Roman*)}]
  \item \emph{Ray-separable closure:} $\mathcal S_2=\varnothing$.  Every
  nonzero positive-arity all-left boundary tensor is completely
  decomposable; a nullary tensor is a scalar.

  \item \emph{Gaussian-mobile closure:} $\mathcal S_2\ne\varnothing$ and
  no monomial flag exists.  Then $M$ has a common matchgate cover of rank
  at most eight, and yields an exactly equivalent common presentation of
  width at most three.  If $|\mathcal S_2|\ge2$, the cover has rank four
  and the width is at most two.

  \item \emph{One-plane flag closure:} a monomial flag $(P,L)$ exists.
  Every nonzero positive-arity all-left boundary tensor has the
  Boolean-core-times-rays form of \cref{sec:statement-supports}, using one
  fixed endpoint-basis isomorphism $\iota:\C^2\to P$ for the entire closure; a nullary
  boundary tensor is a scalar.
\end{enumerate}
The class of languages admitting an alternative~\textup{(III)} presentation
has no uniform bound on minimum exact common width.  The rank,
essentiality, and connected two-center properties in \cref{thm:main} may
hold simultaneously.
\end{theorem}

\begin{remark}[Structural trichotomy versus width characterization]
\label[remark]{rem:not-pointwise-dichotomy}
The alternatives are presentation-relative because $P^+$ and $P^-$ are
defined through the displayed base.  The bounded conclusions in
alternative~\textup{(II)} are nevertheless language-level: they construct
a smaller exact presentation of the same labels.  ``No uniform bound''
quantifies over the family $(\cF_k\mid\{B_k\})$; every individual finite
language has finite width.  A flag is necessary for our constructed family
to escape Gaussian collapse but is not sufficient for large width, since
small-width flag presentations also exist.  Alternative~\textup{(I)} also
has no width classification.  A pointwise characterization requires a
language invariant that captures the surviving right-side data
independently of any chosen flag presentation.
\end{remark}

\Cref{thm:main} is proved in
Sections~\ref{sec:matchgate-varieties}--\ref{sec:positive-boundary},
\cref{cor:no-domain-bound} in \cref{sec:arbitrary-equivalence}, and
\cref{thm:structural-trichotomy} in \cref{sec:positive-boundary}.

\section{Technical Preliminaries}
\label{sec:technical-preliminaries}

\subsection{Sets, fields, tensors, and ranks}
\label{sec:tensor-preliminaries}

We write $\C$, $\Q$, and $\mathbb Z$ for the fields of complex and rational
numbers and the ring of integers, respectively.  Unless another field is
named, all vector spaces, tensor products, spans, and matrix ranks are over
$\C$.  We use
\[
  \mathbb N=\{1,2,\ldots\},\qquad
  \mathbb Z_{\ge0}=\{0,1,2,\ldots\},\qquad
  \mathbb Z_{>0}=\mathbb N.
\]
For a positive integer $a$, write $[a]=\{1,\ldots,a\}$, and put
$[0]=\varnothing$.  The Boolean domain is $\bits$, and a \emph{qutrit}
means a three-element domain, normally identified with
$D_3=\{0,1,2\}$.  For a subfield $K\subseteq\C$, write
$K^\times=K\setminus\{0\}$ and let $\overline{\Q}$ be the field of complex
algebraic numbers.

For $x\in\bits^s$, define
\[
  \supp(x):=\{i\in[s]:x_i=1\},\qquad
  |x|:=|\supp(x)|,
\]
and call $|x|$ the \emph{Hamming weight}.  The symbol $\oplus$ denotes
bitwise exclusive-or,
$\dist(x,y):=|x\oplus y|$ is Hamming distance, and
$\bar x:=x\oplus1^s$ is the bitwise complement.  For $T\subseteq[s]$,
$\one_T\in\bits^s$ is its indicator word.  The words $0^s$ and $1^s$ are
the all-zero and all-one repetition words; for $s=0$ both equal the empty
word $\epsilon$.  If $I$ is an arbitrary finite index set, $\bits^I$
means the set of functions $I\to\bits$, and $\sigma(i)$ denotes the bit at
$i$.  For a proposition $\mathcal A$, the indicator
$\mathbf 1[\mathcal A]$ is $1$ when $\mathcal A$ is true and $0$ otherwise.
For a coordinate array $A$, we also write
$\supp(A):=\{\omega:A(\omega)\ne0\}$ for its coordinate support; the
argument makes this distinct from the support of a bit word.

If $D$ is a finite set, a \emph{signature of arity $a$} over $D$ is a
function $F\colon D^a\to\C$, and we write $\arity(F)=a$.  A nullary
signature is a complex scalar.  A signature is \emph{rational} or
\emph{integer-valued} when all of its coordinates lie in $\Q$ or
$\mathbb Z$, respectively.  All signature families below are finite
labelled families: labels remain distinct even if two labels carry the same
coordinate tensor.  A \emph{primitive signature} is a tensor attached to a
single labelled vertex, as opposed to a tensor obtained by contracting a
multi-vertex gadget; its argument positions are its \emph{ports}.

Let $V_D=\C^D$ with its standard basis.  Write $e_a$ for the coordinate
vector indexed by $a\in D$ and $e_a^*$ for its dual covector.  More
generally, $e_x$ denotes the coordinate vector indexed by a word $x$, so
$e_x(y)$ is $1$ for $y=x$ and $0$ otherwise.  The tensor product
$u_1\otimes\cdots\otimes u_a$ is the multilinear array whose coordinates
are the products of the corresponding coordinates; $M^{\otimes a}$ is the
$a$-fold Kronecker power applying $M$ independently to the $a$ ordered
ports.  The zeroth tensor power is the scalar identity on $\C$, and an empty
tensor product equals $1$.  The symbols $\Span$, $\operatorname{im}$,
$\oplus$, and $\bigoplus$ denote linear span, image, direct sum, and an
iterated direct sum.  We write $A^{\trans}$ for ordinary transpose, with no complex
conjugation, and $\row(A)$, $\operatorname{col}(A)$, and $\rank A$ for the
row space, column space, and matrix rank.

For tensor-variance bookkeeping, we regard a signature placed on the left
side of a bipartite instance as an element of $(V_D^*)^{\otimes a}$ and one
placed on the right as an element of $V_D^{\otimes a}$.  Both have the same
coordinate-array notation, and all pairings in this paper are multilinear
over $\C$, never Hermitian.  A \emph{ray}, \emph{line}, and \emph{plane}
mean respectively a one-dimensional subspace (the first two words are
synonymous here) and a two-dimensional subspace.  If $K\le U$ are
finite-dimensional, $\operatorname{codim}_U K:=\dim U-\dim K$.

For $a\ge1$ and $j\in[a]$, the \emph{one-versus-rest flattening}, also
called the \emph{mode-$j$ flattening}, of $F$ is the matrix
\begin{equation}
  \Mat_j(F)\in\C^{|D|\times |D|^{a-1}},
  \qquad
  \Mat_j(F)_{u,\mathbf v}
  =F(v_1,\ldots,v_{j-1},u,v_j,\ldots,v_{a-1}).
  \label{eq:flattening-definition}
\end{equation}
Its image
\begin{equation}
  S_j(F):=\operatorname{im}\Mat_j(F)
  \label{eq:mode-support-definition}
\end{equation}
is the \emph{mode support} or \emph{port support}, and
$\rank\Mat_j(F)=\dim S_j(F)$ is the corresponding \emph{mode rank}.  With
the variance convention above, a left mode support is viewed in $V_D^*$
and a right mode support in $V_D$.  A port is \emph{deficient} when its
mode rank is less than $|D|$.  When $|D|=n$, we call $F$ \emph{full rank}
if $\rank\Mat_j(F)=n$ for some $j$, and \emph{non-full-rank} if every port
is deficient.  This is Chen's domain-signature notion of full rank
\cite{Chen2016}; it is unrelated to the matrix rank of a holographic base.
A presentation is \emph{port-deficient} when every port of every primitive
signature in both induced domain families is deficient.

A \emph{minimal mode decomposition} at port $j$ is an expression
\[
  F=\sum_{a=1}^{d}u_a\otimes F_a,
  \qquad d=\rank\Mat_j(F),
\]
after placing port $j$ first in the notation, with both $(u_a)$ and
$(F_a)$ linearly independent.  It is precisely a rank factorization of the
mode-$j$ flattening.

A family of left signatures is \emph{support-essential} when the sum of
all of its mode supports is $V_D^*$; the analogous condition for a right
family uses $V_D$.  A bipartite problem, or a presentation of it, is
\emph{bi-essential} when both induced domain families are
support-essential.  Zero tensors contribute the zero subspace to these
sums.

A nonzero positive-arity tensor is \emph{completely decomposable}, or a
\emph{pure tensor}, if it is an elementary tensor
$u_1\otimes\cdots\otimes u_a$; equivalently, every mode rank is one.  For a
bipartition $A\sqcup B=[a]$, the \emph{Schmidt rank across $A\mid B$} is
the matrix rank obtained by flattening the $A$ coordinates against the $B$
coordinates.  Thus Schmidt rank one is exactly factorization as one tensor
on $A$ times one tensor on $B$.  An \emph{outer product} is such a
rank-one bipartite tensor.
A \emph{slice} of a tensor is the lower-arity tensor obtained by fixing a
specified subset of its coordinates.

Finally, $u_k=\Omega(v_k)$ means that there are constants $c>0$ and $k_0$
such that $|u_k|\ge c|v_k|$ for every $k\ge k_0$.  For
$q=p/r\in\Q$ in lowest terms, with $r>0$, define its \emph{height} by
\[
  H(q):=\max\{|p|,r\}.
\]
The height of a finite rational table is the maximum height of its entries.
If $H$ is this table height, then $1+\log_2 H$ controls the maximum entry
bit length up to constant factors.  We call an indexed family of tables
\emph{entrywise closed-form} when one uniform formula gives each entry
directly from its index, rather than selecting the table by existence or
enumeration.  A table is \emph{everywhere nonzero} or \emph{nowhere zero}
when none of its entries is zero.

\subsection{Planar Holant networks and gadgets}
\label{sec:holant-preliminaries}

Let $\cF$ and $\cG$ be finite labelled signature families over the same
domain $D$.  We call the ordered pair $\cF\mid\cG$ a \emph{bipartite
signature language} or \emph{planar Holant problem}.  A planar bipartite
$\cF\mid\cG$ instance consists of a planar
embedded bipartite graph $(U,V,E)$, a label from $\cF$ on every vertex of
$U$, and a label from $\cG$ on every vertex of $V$.  At each vertex of
positive degree we specify a first incident edge and then follow the cyclic
order induced by the embedding; a nullary vertex has the unique empty port
order.  Thus every vertex has a linear port order compatible with its
planar cyclic order.  A vertex label must have arity equal to the degree of
that vertex, and its arguments follow this ordered list of ports.  Here a
planar embedding is a fixed crossing-free embedding; it determines an outer
face and a cyclic order at every vertex.  We write $E(w)$ for the resulting
linearly ordered list of edges incident with $w$.
The value of the instance $I$ is
\begin{equation}
  \Hol_I(\cF\mid\cG)
  =
  \sum_{\sigma\colon E\to D}
    \prod_{u\in U} F_u(\sigma|_{E(u)})
    \prod_{v\in V} G_v(\sigma|_{E(v)}).
  \label{eq:planar-holant}
\end{equation}
Here $F_u$ and $G_v$ denote the coordinate tensors assigned by the labels,
and $\sigma|_{E(w)}$ is the tuple of incident edge values read in the
specified linear port order.  Formula~\eqref{eq:planar-holant} is the full
tensor contraction of the network.

A \emph{planar gadget} is the same kind of finite tensor network embedded in
a closed disc, with some uncontracted half-edges ending on the boundary of
the disc; these are its \emph{dangling} or \emph{boundary edges}.  Their
cyclic order on the boundary, together with a marked first edge, is the
specified linear boundary order.  Contracting all internal edges produces
the \emph{boundary tensor}, whose ports occur in that order.  If $e$ is a
boundary edge, $j(e)$ denotes its position in this order.  The vertex tensors
are the gadget's primitive signatures.  A
gadget is \emph{connected} when its underlying graph after deleting the
dangling half-edges is connected.  It is \emph{all-left} when every
dangling edge is incident with a left vertex.  The \emph{all-left gadget
closure} is the collection of boundary tensors of actual connected
all-left gadgets; it is not their linear span.  A \emph{star gadget} has
one central vertex; each central port is either contracted with a unary
leaf or left as a dangling edge.  It is \emph{closed} when it has no
dangling edges, equivalently, when its boundary arity is zero and every
central port is contracted with a unary leaf.  A \emph{pin} is a unary
coordinate signature $\delta_a=e_a^*$ on the left, or $e_a$ on the right;
\emph{pinning a port} means contracting it with that unary tensor.  Ordered
planar composition means gluing compatible consecutive boundary ports
without changing their inherited cyclic orders.  The Boolean
\emph{disequality} signature is $\mathrm{NEQ}(x,y)=1$ for $x\ne y$ and $0$
otherwise.

\subsection{Planar matchgates, Pfaffians, and ordered operations}
\label{sec:matchgate-preliminaries}

A \emph{perfect matching} of a finite graph is a set of edges incident with
every vertex exactly once.  For an edge-weighted graph $\mathcal N$, let
$w(e)$ denote the weight of edge $e$ and put
\begin{equation}
  \PerfMatch(\mathcal N)
  :=\sum_{P\text{ a perfect matching of }\mathcal N}
       \prod_{e\in P}w(e),
  \label{eq:weighted-perfect-matching}
\end{equation}
with value $0$ if no perfect matching exists and value $1$ on the empty
graph.  A \emph{planar matchgate} $\mathcal M$ is a finite edge-weighted
planar graph with \emph{external vertices} $v_1,\ldots,v_s$ occurring in
that cyclic order on its outer face, with $v_1$ marked as first; all other
vertices are \emph{internal}.  Under the deletion
convention, its $s$-bit signature $G\colon\bits^s\to\C$ is
\begin{equation}
  G(x)=
  \PerfMatch\!\left(
    \mathcal M-\{v_i:x_i=1\}
  \right),
  \label{eq:matchgate-signature}
\end{equation}
where deletion removes the indicated external vertices and all incident
edges.  We fix this convention throughout.  The other common convention is
obtained by globally complementing every index bit.  Let $\MG_s$ denote the
set of \emph{exact} $s$-bit planar-matchgate signatures in the stated
external order: ``exact'' means equality with the finite weighted graph's
coordinates, not approximation or a limiting construction.  We set
$\MG_0=\C$.

For $s\ge0$, let
\[
  \rho_s(x_1,\ldots,x_s)=(x_s,\ldots,x_1),
  \qquad \rho_0(\epsilon)=\epsilon.
\]

\begin{lemma}[Reflection reversal]
\label[lemma]{lem:reflection-reversal}
If $G\in\MG_s$, then the reversed signature
\[
  G^{\mathrm{rev}}(x):=G(\rho_s(x))
\]
also belongs to $\MG_s$.
\end{lemma}

\begin{proof}
For $s=0$ there is nothing to prove.  Reflect an exact planar realization
of $G$ in the plane and mark the reflected copy of the old last external
vertex as the new first external vertex.  The reflected external order is
the reverse of the old order, while the weighted graph and every
perfect-matching weight are unchanged.  Its deletion signature in the new
order is therefore $G\circ\rho_s$.
\end{proof}

The \emph{matchgate identities}, also called the quadratic
Grassmann--Pl\"ucker relations in these coordinates, are the following
equations.  For $\alpha,\beta\in\bits^s$, list
$\supp(\alpha\oplus\beta)=\{p_1<\cdots<p_\ell\}$; then
\begin{equation}
  \sum_{j=1}^{\ell}(-1)^j
  G(\alpha\oplus\one_{\{p_j\}})
  G(\beta\oplus\one_{\{p_j\}})=0.
  \label{eq:matchgate-identities}
\end{equation}
They are necessary and sufficient for exact planar-matchgate realizability
and imply parity~\cite{CaiGorenstein2014,CaiChoudharyLu2009}.

A matrix $A$ is \emph{skew-symmetric} when $A^{\trans}=-A$.  If
$S=\{i_1<\cdots<i_{2m}\}$ has even size, its principal submatrix is $A[S]$
and its \emph{Pfaffian} is
\begin{equation}
  \Pf(A[S])
  :=\sum_{\pi}\sgn(\pi)
       \prod_{q=1}^{m}A_{u_q,v_q}.
  \label{eq:pfaffian-definition}
\end{equation}
Here $\pi$ ranges over partitions of $S$ into pairs.  For each $\pi$, write
each pair as $\{u_q,v_q\}$ with $u_q<v_q$, and index the pairs so that
$u_1<\cdots<u_m$.  The increasing order on $S$ is the displayed sequence
$(i_1,\ldots,i_{2m})$.  Let $\sigma_\pi$ be the unique permutation of
$[2m]$ satisfying
\[
  (u_1,v_1,\ldots,u_m,v_m)
  =(i_{\sigma_\pi(1)},i_{\sigma_\pi(2)},\ldots,i_{\sigma_\pi(2m)}).
\]
Then
\[
  \sgn(\pi):=\sgn(\sigma_\pi)=(-1)^{N_\pi},
\]
where $N_\pi$ is the number of pairs of positions $a<b$ for which the
$a$th entry of $(u_1,v_1,\ldots,u_m,v_m)$ is larger than its $b$th entry.
For example, if $S=\{1,2,3,4\}$ and
$\pi=\{\{1,3\},\{2,4\}\}$, then the associated sequence
$(1,3,2,4)$ has one inversion, so $\sgn(\pi)=-1$.  Each product is a
\emph{perfect-matching monomial}.  We use $\Pf(A[\varnothing])=1$.

For skew-symmetric matrices $A$ and $B$ of even orders, their block direct
sum in concatenated coordinate order satisfies
\[
  \Pf(A\oplus B)=\Pf(A)\Pf(B).
\]
More generally, if $P_\sigma$ is the permutation matrix of a coordinate
permutation $\sigma$, then, for every even-order skew-symmetric matrix $C$,
\[
  \Pf(P_\sigma C P_\sigma^{\trans})
  =\det(P_\sigma)\Pf(C)
  =\sgn(\sigma)\Pf(C).
\]
Swapping consecutive coordinate chunks of sizes $p$ and $q$ consists of
$pq$ transpositions and therefore contributes the sign $(-1)^{pq}$.
Consequently, permuting even-sized chunks introduces no sign.

We use the following two standard consequences of the identities.

\begin{enumerate}
  \item Every nonzero $G\in\MG_s$ has one parity: all nonzero coordinates
  $G(x)$ have the same parity of $|x|$.

  \item For a skew-symmetric matrix $A\in\C^{s\times s}$, define its
  \emph{principal-Pfaffian signature}
  \begin{equation}
    \Gamma_A(x)=
    \begin{cases}
      \Pf(A[\supp(x)]),&|x|\text{ is even},\\
      0,&|x|\text{ is odd},
    \end{cases}
    \qquad \Pf(A[\varnothing])=1,
    \label{eq:principal-pfaffian-tensor}
  \end{equation}
  Then $\Gamma_A\in\MG_s$.  By sufficiency of the matchgate identities,
  $\Gamma_A$ has an exact planar realization in the stated external order;
  $A$ is a Pfaffian parameter matrix and need not itself be planar.
  To reconcile conventions explicitly, if
  $\widetilde\Gamma_A(x):=\Pf(A[[s]\setminus\supp(x)])$, then
  $\Gamma_A(x)=\widetilde\Gamma_A(\bar x)$.  The support- and
  deletion-Pfaffian conventions therefore differ by the one global bit
  complement already fixed above.  The matchgate identities are invariant
  under this common XOR shift of all coordinate indices, so
  $\Gamma_A$ satisfies them in the same coordinate order and the
  sufficiency theorem realizes it exactly in the deletion convention.
\end{enumerate}

The all-zero signature belongs to $\MG_s$: one may adjoin an isolated
internal vertex, which destroys every perfect matching.  Every standard
basis vector $e_x$ belongs to $\MG_s$ as a tensor product of unary pins.  If
$G\in\MG_s$ and $\lambda\in\C$, then $\lambda G\in\MG_s$; take the
disjoint union with an internal edge of weight $\lambda$.  More generally,
exact matchgates are closed under disjoint union/tensor product, pinning an
external bit, and ordered planar contraction: the last operation glues
compatible consecutive external vertices, sums over their common Boolean
values, and retains the inherited order on the unglued vertices.  These are
the closure operations invoked later.

\begin{lemma}[Field-preserving matchgate realization]
\label[lemma]{lem:field-preserving-realization}
Let $K\subseteq\C$ be a subfield, and fix an order of $s$ coordinates.  If
$G\in K^{\bits^s}$ satisfies the matchgate identities in that order, then
$G$ is exactly the deletion signature of a planar matchgate whose edge
weights lie in $K$ and whose external vertices occur in the prescribed
order.
\end{lemma}

\begin{proof}
For a nonzero signature this is the field-preserving sufficiency
construction for the matchgate identities of Cai and
Gorenstein~\cite[Section~5 and Theorem~5.1]{CaiGorenstein2014}.  Their
normalization and restoration steps use only field operations, and their
planar crossover gadget---the fixed matchgate implementing a signed exchange
of two ordered wires---uses weights in $\{+1,-1\}$.  Their construction
places the labels $1<\cdots<s$ clockwise; naming those labels by the fixed
coordinate order gives the prescribed external order here.  The zero signature is
realized over $K$ by including an isolated internal vertex.  Nullary
signatures are scalars and are realized by a single internal edge of the
given weight (with the zero case handled as above).
\end{proof}

\subsection{Matchgate matrices, covers, and fermionic geometry}
\label{sec:matchgate-matrix-preliminaries}

Let $G\in\MG_{r+t}$ have its first $r$ ports designated as inputs and its
last $t$ as outputs.  Its \emph{$r$-input, $t$-output exact matchgate matrix}
is the $2^r$-by-$2^t$ matrix $P$ defined by
\begin{equation}
P(x,y):=G(x,\rho_t(y)),
  \qquad x\in\bits^r,\quad y\in\bits^t,
  \label{eq:matchgate-matrix-convention}
\end{equation}
where $\rho_t$ is the fixed planar reversal defined in
\cref{lem:reflection-reversal}.
With this convention, gluing an ordered output block to an ordered input
block is ordinary
matrix multiplication.  A square such matrix is an \emph{input} or
\emph{output matchgate transformation} according to the side on which it
is composed; it is \emph{invertible} in the ordinary linear-algebraic
sense.  The inverse of an invertible exact matchgate transformation is
again an exact matchgate transformation, by the ordered canonical-form
theorem~\cite{CaiChoudharyLu2009,Xia2016}.  Here an \emph{ordered canonical
form} means the normalization of an exact matchgate matrix, by invertible
input and output matchgate transformations, to disjoint active
input--output wires with all inactive modes pinned; every later use states
the particular consequence needed.  A Boolean \emph{wire} or
\emph{mode} is one binary tensor index;
an ordered consecutive group of wires is a \emph{block}.

The matchgate matrix rank theorem says that if $P$ is a nonzero
$a$-input, $b$-output exact matchgate matrix, then
\[
  \rank P=2^d
\]
for a unique integer $d$ with $0\le d\le\min\{a,b\}$.  In the ordered
canonical form, $d$ is the number of active input--output wires
\cite[Theorem~4.1]{Xia2016}.  Conversely, the canonical matchgate with $d$
active wires and all remaining modes pinned shows that, for fixed $a,b$,
every $d$ in this range occurs.  Thus, when the input and output arities are
unrestricted, $d$ may be any nonnegative integer.  For $0\le r\le t$, the
\emph{standard inclusion} $J_{r,t}$ is the exact matchgate matrix that
copies the $r$ input bits to $r$ active output wires and pins the remaining
$t-r$ inactive output wires to zero.  In the fixed matrix convention,
\[
  J_{r,t}(x,y)=\mathbf 1[y=(x,0^{t-r})],
  \qquad J_{r,t}J_{r,t}^{\trans}=I_{2^r}.
\]
In its planar realization, the first $t-r$ physical output ports are
zero-pinned and $r$ nested unit edges join input $i$ to the physical output
port carrying logical bit $y_i$.  Thus the physical output order is
$y_t,\ldots,y_1$, exactly as required by
\eqref{eq:matchgate-matrix-convention}.
A fixed-mode \emph{flip} exchanges the two Boolean values on that mode and
is itself an invertible matchgate transformation: replace the corresponding
unit equality edge in the nested identity realization by a two-edge unit
path with one internal vertex.  Its two-terminal deletion signature is
$1$ exactly on $01$ and $10$, so its ordinary matrix is the bit flip on
that mode.

Following Xia, if $M\in\C^{n\times2^t}$ has rank $2^r$, a
\emph{full-rank matchgate realizer} of $M$ is a matrix
$A\in\C^{2^s\times n}$, for some $s\ge0$, such that $AM$ is an
$s$-input, $t$-output exact matchgate matrix of rank $2^r$
\cite[Definition~5.1]{Xia2016}.  More generally, a
\emph{rank-$2^r$ common matchgate cover} of a width-$t$ base $M$ is an
$r$-input, $t$-output exact matchgate matrix $P$ of rank $2^r$ satisfying
\begin{equation}
  \row(M)\subseteq\row(P).
  \label{eq:common-cover-definition}
\end{equation}
Its \emph{cover width} is $r$.  It is called common because the same row
space covers the single base shared by every label.  A \emph{Gaussian
cover} is the same object; ``Gaussian'' here is the standard fermionic
synonym for matchgate/Pfaffian structure, not a probabilistic term.
A \emph{rank-$2^r$ Gaussian hull matrix} of subspaces
$R_1,\ldots,R_m$ is such an exact matchgate matrix whose row space contains
$R_1+\cdots+R_m$; its row space is the corresponding \emph{Gaussian hull}.
When that row space contains $\row(M)$, the matrix is a common cover of $M$.

For parity bookkeeping, put
\[
  \mathcal E_t:=\Span\{e_x:|x|\text{ is even}\},\qquad
  \mathcal O_t:=\Span\{e_x:|x|\text{ is odd}\}.
\]
A subspace is \emph{parity-invariant} when it is the direct sum of its
intersections with $\mathcal E_t$ and $\mathcal O_t$.  If $R$ is the row
space of a rank-two exact matchgate matrix, the canonical form gives
\[
  R=R^+\oplus R^-,\qquad
  R^+:=R\cap\mathcal E_t,\quad R^-:=R\cap\mathcal O_t,\qquad
  \dim R^+=\dim R^-=1.
\]
A nonzero vector is \emph{parity-homogeneous} when it lies entirely in
$\mathcal E_t$ or entirely in $\mathcal O_t$.
The two lines $R^+,R^-$ are its \emph{parity endpoints}.  When a domain
plane $P$ maps to $R=PM$, their pullbacks in $P$ are denoted $P^+,P^-$.

We now fix the pure-spinor terminology used in the support argument.  Let
$f_1,\ldots,f_t$ be the standard basis of $W=\C^t$, and identify
$\C^{\bits^t}$ with the exterior algebra $\Lambda^\bullet W$ by
\[
  e_x\longleftrightarrow
  f_{i_1}\wedge\cdots\wedge f_{i_m}
  \quad\text{when }\supp(x)=\{i_1<\cdots<i_m\}.
\]
Thus $e_\varnothing=e_{0^t}$ is the \emph{vacuum}, a vector supported on
words of weight one is a \emph{one-particle vector}, and $\wedge$ is the
alternating exterior product.  The space $W\oplus W^*$ carries the
symmetric bilinear form
\[
  \langle w+\phi,w'+\phi'\rangle
  :=\phi(w')+\phi'(w)
\]
and acts on $\Lambda^\bullet W$ by the Clifford action
\[
  (w+\phi)\cdot u:=w\wedge u+\iota_\phi u,
\]
where $\iota_\phi$ is contraction by $\phi$.  For nonzero $u$, define its
\emph{annihilator}
\begin{equation}
  \Ann(u):=\{z\in W\oplus W^*:z\cdot u=0\}.
  \label{eq:pure-spinor-annihilator}
\end{equation}
A subspace $K\le W\oplus W^*$ is \emph{isotropic} when the displayed form
vanishes on $K\times K$; it is maximal isotropic when it has the maximum
  possible dimension $t$.  A nonzero $u$ is a \emph{pure spinor} when
$\Ann(u)$ is maximal isotropic.  In the coordinates above, satisfying the
matchgate identities in the fixed ordered basis is equivalent to being a
pure spinor.  Together with \cref{lem:field-preserving-realization}, a
nonzero vector is therefore an exact matchgate vector if and only if it is
a pure spinor
\cite{CaiChoudharyLu2009,CaiGorenstein2014,Manivel2009}.  For a vector space
$U$, its projective space $\mathbb P(U)$ is the
set of one-dimensional subspaces; $[u]$ denotes the point $\C u$.  We call
$\C u$ a \emph{pure matchgate ray} and define
$\Ann(\C u):=\Ann(u)$, which is independent of the chosen nonzero generator.
The projective set of pure matchgate rays of one fixed parity is the
\emph{spinor variety}; a \emph{pure pencil} is a projective line contained
in that variety.  A \emph{common annihilator} of several pure vectors or
rays is the intersection of their annihilators.  The contractions
$\iota_{f_i^*}$ are the \emph{standard annihilation operators}; their common
kernel on $\Lambda^\bullet W$ is the intersection of their ordinary linear
kernels.

A linear \emph{isometry} of $W\oplus W^*$ is an invertible linear map that
preserves the displayed symmetric bilinear form.  The \emph{Witt extension
theorem} says that an isometry between two isotropic subspaces extends to an
isometry of the whole nondegenerate space.  The group of all such
isometries is denoted $\mathrm O(W\oplus W^*)$, and its determinant-one subgroup is
$\mathrm{SO}(W\oplus W^*)$.  The \emph{Pin/Spin (Clifford)
lift} is the induced invertible action on $\Lambda^\bullet W$ intertwining
the Clifford action; a determinant-reversing isometry uses the Pin action,
or equivalently one fixed-mode flip in addition to a Spin action.  Thus, if
$g\in \mathrm O(W\oplus W^*)$ and $U_g$ is a lift, then, up to the immaterial
nonzero scalar in the choice of lift,
\begin{equation}
  U_g\bigl(z\cdot u\bigr)=(gz)\cdot U_g(u).
  \label{eq:clifford-intertwining}
\end{equation}
The next lemma supplies the ordered-planar realization statement needed
later; it is not being folded into the terminology.

\begin{lemma}[Ordered transpose and Clifford-lift realization]
\label[lemma]{lem:ordered-clifford-realization}
The following statements use the matrix convention
\eqref{eq:matchgate-matrix-convention}.
\begin{enumerate}[label=\textup{(\alph*)}]
  \item If $P$ is an $r$-input, $t$-output exact matchgate matrix, then
  $P^{\trans}$ is a $t$-input, $r$-output exact matchgate matrix.

  \item For $g\in \mathrm O(W\oplus W^*)$, the ordinary coefficient matrix of any
  Pin lift $U_g$ in the ordered exterior-monomial basis is an invertible
  $t$-input, $t$-output exact matchgate matrix.  The same is true of
  $U_g^{-1}$.
\end{enumerate}
\end{lemma}

\begin{proof}
For~\textup{(a)}, let $G$ realize $P$, so
$P(x,y)=G(x,\rho_t(y))$.  If $r+t=0$, transposition changes nothing and the
claim is immediate.  Otherwise reflect the embedded graph and mark the reflected
copy of the old last external port as the new first port.  In the resulting
prescribed order its signature is
\[
  H(y,z):=G\bigl(\rho_r(z),\rho_t(y)\bigr).
\]
Reflection preserves planarity, weights, and deletion-signature values.
Moreover,
\[
  H\bigl(y,\rho_r(x)\bigr)
  =G\bigl(x,\rho_t(y)\bigr)=P(x,y)=P^{\trans}(y,x),
\]
so the reflected matchgate realizes exactly $P^{\trans}$.  Notice that this
argument does not assert that a bare unsigned wire reversal is a matchgate
transformation.

For~\textup{(b)}, define the coefficient convention by
\[
  U_g(e_y)=\sum_{x\in\bits^t}(U_g)_{x,y}e_x.
\]
Take ordered copies $W_{\rm in},W_{\rm out}$ of $W$ and, in the ordered
graded-tensor identification
$\Lambda(W_{\rm in}\oplus W_{\rm out})
 \cong\Lambda W_{\rm in}\,\widehat\otimes\,\Lambda W_{\rm out}$, put
\[
  \Omega_t:=\sum_{y\in\bits^t}
       e_y^{\rm in}\wedge e_{\rho_t(y)}^{\rm out}.
\]
Here $\widehat\otimes$ denotes the graded tensor product, with all input
modes ordered before all physical output modes.
This is the signature of the nested-unit-edge realization of
$J_{t,t}=I_{2^t}$ described above, and hence is a pure spinor.  First assume
$g\in \mathrm{SO}(W\oplus W^*)$ and take its even Spin lift.  In the ordered
factorization above its action on doubled spinors is literally
$U_g\widehat\otimes I$: input creation and contraction act on the first
factor, while contraction of an input covector with the second factor
vanishes.  Since the lift is even, it intertwines the doubled Clifford
action with $g\oplus I$.  Therefore
\[
  \Ann\bigl((U_g\widehat\otimes I)\Omega_t\bigr)
  =(g\oplus I)\Ann(\Omega_t),
\]
which is again maximal isotropic.  Consequently
\[
  \Omega_g:=(U_g\widehat\otimes I)\Omega_t
\]
is pure and therefore satisfies the matchgate identities.  Its ordered
coordinates obey
\[
  \Omega_g\bigl(x,\rho_t(y)\bigr)=(U_g)_{x,y},
\]
so \cref{lem:field-preserving-realization} realizes the ordinary matrix of
$U_g$ as an exact planar matchgate in precisely the prescribed input/output
order.

If $\det g=-1$, fix a standard determinant-reversing mode reflection $s$
with $s^2=I$
whose lift is the exact fixed-mode flip recorded above, and put $h=sg\in
\mathrm{SO}(W\oplus W^*)$.  Up to a nonzero scalar, a Pin lift of $g=s h$ is the
ordered composition of that fixed-mode flip with the Spin lift of $h$.
Exact matchgate matrices are closed under this ordinary planar composition
and under nonzero rescaling.  This proves the assertion for every Pin lift.
Applying the same argument to $g^{-1}$ proves it for $U_g^{-1}$.
\end{proof}

An exact-matchgate row space is also called a \emph{Gaussian row space},
and an invertible matchgate transformation a \emph{Gaussian
transformation}.  In the Grassmann description, symbols such as
$\xi_1,\ldots,\xi_s,\eta_1,\ldots,\eta_t$ are anticommuting input and output
variables.  A \emph{pivot chart} chooses a nonzero matrix coordinate, uses
fixed input/output flips to place it at the vacuum coordinate, and divides
by that coordinate; \emph{pivot-normalized} means that this pivot is $1$.
In such a chart a matchgate matrix has a skew quadratic
parameter matrix in the standard Pfaffian/Grassmann chart
\cite{CaiChoudharyLu2009,Xia2016}
\[
  \mathcal A=
  \begin{pmatrix}A_{\rm in}&B\\-B^{\trans}&D_{\rm out}\end{pmatrix}
\]
and coefficient-generating polynomial
\[
  \exp\!\left(\frac12\xi^{\trans}A_{\rm in}\xi
    +\xi^{\trans}B\eta
    +\frac12\eta^{\trans}D_{\rm out}\eta\right).
\]
The variables anticommute, so the exponential is a finite exterior-algebra
polynomial; the coefficient of each input--output monomial is the
corresponding matrix entry.  If the pivot output row is the vacuum, all its
two-output coordinates vanish and hence $D_{\rm out}=0$.  Multiplication on
the input Grassmann algebra by the inverse input-only exponential removes
$A_{\rm in}$ by ordinary invertible row operations.  More explicitly, for
$q(\xi)=\frac12\xi^{\trans}A_{\rm in}\xi$, the map
\[
  L_{-q}:p(\xi)\longmapsto \exp(-q(\xi))\wedge p(\xi)
\]
is an invertible linear map with inverse $L_q$.  Applied coefficientwise in
$\eta$, its ordinary coefficient matrix left-multiplies the matchgate
matrix and hence leaves its row space unchanged.  No exact-matchgate claim
about $L_{-q}$ is needed: since the exterior elements involved are even,
\[
  \exp(-q)\wedge\exp(q+\xi^{\trans}B\eta)
  =\exp(\xi^{\trans}B\eta),
\]
and the matrix on the right is independently exact by the same
principal-Pfaffian chart.  Its coefficient matrix, up to the fixed signs
dictated by \eqref{eq:matchgate-matrix-convention}, is the
\emph{exterior-compound matrix}
\[
  \bigoplus_{q=0}^{\min(s,t)}\wedge^qB.
\]
If
$d=\rank B$, then
\begin{equation}
  \rank\!\left(\bigoplus_{q=0}^{\min(s,t)}\wedge^qB\right)
  =\sum_{q=0}^{d}\binom dq=2^d.
  \label{eq:exterior-compound-rank}
\end{equation}
The nontrivial Clifford-lift realization has been isolated in
\cref{lem:ordered-clifford-realization}; the remaining geometric inputs are
stated with citations where they enter the proofs.

\subsection{Algebraic geometry and transcendence conventions}
\label{sec:algebraic-preliminaries}

For a field $K\subseteq\C$, affine $N$-space $\mathbb A_K^N$ has
coordinate ring $K[Y_1,\ldots,Y_N]$, the polynomial ring of its coordinate
functions.  An \emph{ideal} is an additive subgroup of a ring closed under
multiplication by arbitrary ring elements.  If
$I\subseteq K[Y_1,\ldots,Y_N]$ is an ideal, write
\[
  V_{\C}(I):=\{y\in\C^N:p(y)=0\text{ for every }p\in I\}.
\]
A set $X$ cut out in $\mathbb A_K^N$ by an ideal
$I\subseteq K[Y_1,\ldots,Y_N]$ is called a $K$-defined
\emph{affine algebraic set}; its complex points are
$X(\C):=V_{\C}(I)$.  A subset of $\C^N$ is \emph{$K$-defined Zariski
closed} if it equals $X(\C)$ for such an $X$.  The \emph{Zariski closure}
$\overline{S}^{\mathrm{Zar}}$ is the smallest Zariski-closed set containing
$S$; $S$ is \emph{Zariski dense} when that closure is the whole ambient
space, and a closed subset is \emph{proper} when it is not the whole
ambient space.  A closed set is
\emph{irreducible} when it is not a union of two proper closed subsets, and
an \emph{irreducible component} is a maximal irreducible closed subset.  We
use \emph{affine algebraic set} and \emph{affine variety} interchangeably
and do not require irreducibility.

The \emph{dimension} of an affine algebraic set is its Krull dimension: the
largest length of a strict chain of irreducible closed subsets, equivalently
the maximum dimension of its irreducible components; in particular
$\dim\mathbb A_K^N=N$.  A \emph{principal open set} has the form
$D(f)=\{x:f(x)\ne0\}$, and
$\mathbb G_m:=\mathbb A^1\setminus\{0\}$.  A \emph{regular map} or
\emph{morphism} has coordinate functions that are regular: locally they are
quotients of polynomials with nonvanishing denominators.  On a closed affine
algebraic set, polynomial coordinate functions give the special case used
below; on $\mathbb G_m$, Laurent polynomials---finite sums of integer powers
of the coordinate---are also regular.  We use the
standard facts that
\[
  \dim\overline{\phi(X)}^{\mathrm{Zar}}\le\dim X
\]
for a morphism $\phi$, that product dimensions add, and that extending the
ground field from $\Q$ to $\C$ does not change dimension.

A polynomial is \emph{homogeneous of degree $d$} when every one of its
monomials has total degree $d$; a homogeneous degree-two polynomial is a
\emph{quadric}.  A projective algebraic set in $\mathbb P(U)$ is the common
zero set of homogeneous polynomials, and its \emph{homogeneous vanishing
ideal} consists of all homogeneous polynomials vanishing on it.  Saying that
this ideal is \emph{generated by quadrics} means that every polynomial in it
is a polynomial combination of degree-two generators.  Restricting a
quadric to a \emph{projective line} $\mathbb P(L)$, where $L$ is a
two-dimensional vector subspace, gives a homogeneous degree-two polynomial in
two line coordinates; three distinct zeros force that restriction to vanish
identically.

If $X=V_{\C}(I)\subseteq\C^m$ and a polynomial map
$\phi:X\to\C^N$ has target coordinates $Y_1,\ldots,Y_N$, its
\emph{graph ideal} is
\[
  J:=I+\langle Y_j-\phi_j(\mathbf U):j\in[N]\rangle
  \subseteq K[\mathbf U,\mathbf Y],
\]
and $J\cap K[\mathbf Y]$ is the corresponding \emph{elimination ideal}.
The elimination theorem identifies the Zariski closure of the image with
$V_{\C}(J\cap K[\mathbf Y])$.  A \emph{Gr\"obner basis} is a generating
set whose leading monomials generate the leading-monomial ideal for a fixed
\emph{term order}, meaning a multiplicative well-order of monomials; the
leading monomial is the largest monomial in that order.  An
elimination-order basis computes this intersection.  We use
only the elementary fact that a Gr\"obner basis over $\Q$ remains one after
extending coefficients to $\C$.

For $S\subseteq\C$, $\Q(S)$ denotes the smallest subfield of $\C$
containing $\Q\cup S$.  Elements $a_1,\ldots,a_m$ are \emph{algebraically
independent over $\Q$} if no nonzero polynomial in
$\Q[X_1,\ldots,X_m]$ vanishes at them.  The \emph{transcendence degree}
$\trdeg_{\Q}K$ is the maximum cardinality of an algebraically independent
subset of $K$; the \emph{coordinate field} of a tensor is the field
generated by all of its coordinates.  The \emph{compositum} of subfields
is the smallest field containing all of them, and transcendence degree is
subadditive under finite composita.  By \emph{finite Zariski avoidance} we
mean choosing a point outside a finite union of proper Zariski-closed
sets.

For the number-theoretic argument, the \emph{square-class group}
$\Q^\times/(\Q^\times)^2$ identifies nonzero rationals whose quotient is a
rational square.  A \emph{multiquadratic field} is a field
$\Q(\sqrt{d_1},\ldots,\sqrt{d_m})$ obtained by adjoining finitely many
square roots.  These terms occur only in the proof of the explicit
transcendental table.

\subsection{Bases and valid common presentations}
\label{sec:base-presentation}

For an integer $t\ge0$, a width-$t$ base from Boolean blocks to a finite
domain $D$ is a matrix
\begin{equation}
  M\in\C^{|D|\times2^t}.
  \label{eq:xia-base}
\end{equation}
Rows are indexed by $D$ and columns by $t$-bit strings.  This is Xia's
orientation~\cite{Xia2016}; Chen uses the transpose orientation
\cite{Chen2016}.

A row of the form $e_b^{\trans}$ encodes its domain state by the Boolean
\emph{block codeword} $b\in\bits^t$.  For several blocks we abbreviate
\[
  e_{x_1,\ldots,x_m}:=e_{x_1}\otimes\cdots\otimes e_{x_m}.
\]
If $G=M^{\otimes b}H$, then $H$ is a \emph{Boolean preimage} of the domain
signature $G$ under $M$.  A \emph{base collapse} is the replacement of a
valid width-$t$ presentation by a valid presentation of the same problem at
smaller width.  A \emph{holographic algorithm} uses such a presentation to
replace the domain tensor network by a Boolean matchgate network and hence
by weighted perfect matching; the present paper studies exact
representability and does not infer an algorithmic running time from width
alone.

We view a left domain signature $F\colon D^a\to\C$ as a row tensor and a
right Boolean signature $H\colon(\bits^t)^b\to\C$ as a column tensor.  We
call $b$ the \emph{block arity} of $H$ and write $\barity(H)=b$; its
ordinary Boolean arity is $bt$.  Their blockwise transforms are
\begin{align}
  \widehat F
  &=F M^{\otimes a},
  &
  \widehat F(x_1,\ldots,x_a)
  &=\sum_{\mathbf i\in D^a}
      F(\mathbf i)\prod_{j=1}^a M_{i_j,x_j},
  \label{eq:left-transform}\\
  G
  &=M^{\otimes b}H,
  &
  G(\mathbf i)
  &=\sum_{x_1,\ldots,x_b\in\bits^t}
      H(x_1,\ldots,x_b)
      \prod_{j=1}^b M_{i_j,x_j}.
  \label{eq:right-transform}
\end{align}
Every $x_j$ is a consecutive $t$-bit block, and the displayed block order
is the external linear order compatible with the planar cyclic order.
This is \emph{block-major order}: list every wire of block $1$, then every
wire of block $2$, and so on.  The traditional names for the transformed
left signatures $FM^{\otimes a}$ and the Boolean right preimages $H$ are
\emph{generators} and \emph{recognizers}, respectively.  \emph{Simultaneous
realizability} means that the same base makes all of these labelled tensors
exact matchgates.  A \emph{full-rank generator} is a generator whose
matrix obtained by choosing one Boolean block as output and all remaining
blocks as inputs has rank equal to $\rank M=2^r$; writing that matrix as
$AM$, its left factor $A$ is then a full-rank matchgate realizer in the sense
of \cref{sec:matchgate-matrix-preliminaries}.  This phrase is used only in
the historical discussion.  A \emph{singleton
symmetric left family} contains one
signature invariant under every permutation of its ports; a
\emph{blockwise-symmetric} Boolean signature is invariant under permutations
of its equal-sized Boolean blocks.  These historical terms occur only in
the literature discussion.
Equations~\eqref{eq:left-transform}--\eqref{eq:right-transform} also fix the tensor
variance used throughout the paper.

For a Boolean column tensor $Q$ on $k$ consecutive $r$-bit blocks and
Boolean row tensors $g_1,\ldots,g_k$ on one $r$-bit block, define the
multilinear, non-Hermitian contraction
\begin{equation}
  \Contr(Q;g_1,\ldots,g_k)
  :=
  \sum_{x_1,\ldots,x_k\in\bits^r}
    Q(x_1,\ldots,x_k)\prod_{j=1}^k g_j(x_j).
  \label{eq:multilinear-contraction}
\end{equation}
No complex conjugation occurs.  For $r\ge1$, the contraction is represented
by the closed block-star tensor network in
\cref{fig:block-star-contraction}.  Each arm contains $r$ distinct ordered
Boolean wires, whose values form
$x_j=(x_{j,1},\ldots,x_{j,r})\in\bits^r$.  These entries are separate
summation indices; in particular, there is no copying constraint
$x_{j,1}=\cdots=x_{j,r}$.  The schematic is drawn for $k\ge3$; when
$k=1$ or $2$, only the corresponding distinct arms remain.

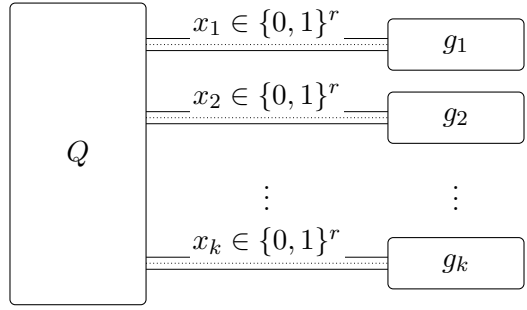
\begin{figure}[htbp]
  \centering
  \begin{tikzpicture}[
      tensor/.style={draw,rounded corners=2pt,fill=white,
        minimum width=1.8cm,minimum height=.68cm,inner sep=2pt},
      centre/.style={tensor,minimum height=4.0cm},
      wire/.style={line width=.35pt}
    ]
    \node[centre] at (-2.5,0) {$Q$};
    \node[tensor] at (2.5,1.45) {$g_1$};
    \node[tensor] at (2.5,.48) {$g_2$};
    \node at (2.5,-.48) {$\vdots$};
    \node[tensor] at (2.5,-1.45) {$g_k$};

    \draw[wire] (-1.6,1.53)--(1.6,1.53);
    \draw[wire,densely dotted] (-1.6,1.45)--(1.6,1.45);
    \draw[wire] (-1.6,1.37)--(1.6,1.37);
    \draw[wire] (-1.6,.56)--(1.6,.56);
    \draw[wire,densely dotted] (-1.6,.48)--(1.6,.48);
    \draw[wire] (-1.6,.40)--(1.6,.40);
    \draw[wire] (-1.6,-1.37)--(1.6,-1.37);
    \draw[wire,densely dotted] (-1.6,-1.45)--(1.6,-1.45);
    \draw[wire] (-1.6,-1.53)--(1.6,-1.53);

    \node[fill=white,inner sep=1pt] at (0,1.73) {$x_1\in\bits^r$};
    \node[fill=white,inner sep=1pt] at (0,.76) {$x_2\in\bits^r$};
    \node at (0,-.48) {$\vdots$};
    \node[fill=white,inner sep=1pt] at (0,-1.17) {$x_k\in\bits^r$};
  \end{tikzpicture}
  \caption{The closed block-star tensor network representing
    $\Contr(Q;g_1,\ldots,g_k)$.  The central tensor $Q$ has $k$ consecutive
    $r$-port blocks, and the $j$th arm connects $g_j$ to block $j$.
    The within-arm order is fixed in the paragraph below.  Each
    three-stroke bundle schematically represents
    the $r$ distinct wires of an arm and is not a literal wire count;
    intermediate arms are suppressed.}
  \label{fig:block-star-contraction}
\end{figure}

For $r=0$, every block is the unique empty word, no wires are present, and
the contraction reduces to $Q\prod_{j=1}^k g_j$.

We use the identity block-order convention in
\eqref{eq:multilinear-contraction}: the $j$th block and its internal
coordinate order in $Q$ contract with $g_j$ in that same order.  There is
one planar-order point to check.  Noncrossing physical gluing reads the
ports of the leaf in the reverse order from the central block.  Realize on
the $j$th arm the reflected exact matchgate
$g_j^{\mathrm{rev}}=g_j\circ\rho_r$, which is exact by
\cref{lem:reflection-reversal}.  If the central coordinate is $x_j$, the
leaf is physically read as $\rho_r(x_j)$ and contributes
$g_j^{\mathrm{rev}}(\rho_r(x_j))=g_j(x_j)$.  Thus the algebraic same-order
contraction in \eqref{eq:multilinear-contraction} is an exact planar
matchgate contraction whenever $Q$ and the $g_j$ are exact matchgates.

In Xia's output-port convention, the same scalar is written, after the
coordinate reindexing
$Q^{\rho}(x_1,\ldots,x_k):=Q(\rho_r(x_1),\ldots,\rho_r(x_k))$, as
\[
  \sum_{x_1,\ldots,x_k\in\bits^r}
    Q^{\rho}(x_1,\ldots,x_k)
    \prod_{j=1}^k g_j(\rho_r(x_j)).
\]
Here $Q^\rho$ is only a fixed renaming of coordinates for comparison of
conventions; we do not assert that an arbitrary blockwise reversal of $Q$
is itself a matchgate operation.  The renaming uses exactly the same
coordinate sets and hence does not alter any coordinate-field or
transcendence-degree bound.  We suppress these fixed permutations below.

We also use partial contraction.  For $J\subseteq[k]$, define the tensor on
the uncontracted blocks by
\begin{equation}
  \Contr_J\bigl(Q;(g_j)_{j\in J}\bigr)
  \bigl((x_\ell)_{\ell\notin J}\bigr)
  :=\sum_{(x_j)_{j\in J}}
      Q(x_1,\ldots,x_k)\prod_{j\in J}g_j(x_j).
  \label{eq:partial-contraction}
\end{equation}
In an expression such as
$\Contr(Q;g_1,g_2,-,\ldots,-)$, each dash marks an uncontracted block and
the surviving blocks retain their original order.  The reflected-leaf
argument above and ordered planar matchgate closure from
\cref{sec:matchgate-preliminaries} show that such a partial contraction is
again an exact matchgate whenever all participating tensors are exact
matchgates.

\begin{definition}[Valid holographic presentation]
\label[definition]{def:valid-presentation}
Let $\cF$ be a finite labelled family of domain-$D$ row signatures, let
$\cH$ be a finite labelled family of Boolean column signatures with
specified block arities, and let $M$ be a width-$t$ base.  The triple
$(\cF,\cH,M)$ is a \emph{valid common planar holographic presentation} if
\[
  F M^{\otimes\arity(F)}\in\MG_{t\arity(F)}
  \quad\text{for every }F\in\cF,
  \qquad
  H\in\MG_{t\barity(H)}
  \quad\text{for every }H\in\cH,
\]
in the prescribed blockwise cyclic orders.  Its induced domain problem is
\[
  \cF\mid M\cH,
  \qquad
  M\cH
  :=\{M^{\otimes\barity(H)}H:H\in\cH\},
\]
with labels inherited from $\cH$.  The same base $M$ is used for every
label on both sides.
\end{definition}

The presentation is \emph{over a subfield $K\subseteq\C$} if every
coordinate of $M$, every domain tensor, and every Boolean preimage lies in
$K$, and the required exact planar matchgates can be realized with edge
weights in $K$.  We use \emph{valid presentation} and \emph{valid common
matchgate presentation} as abbreviations for the object in
Definition~\ref{def:valid-presentation}.  A statement made \emph{labelwise}
means separately for each original label, without merging or replacing the
label family.

Associativity of tensor contraction gives the holographic identity
\begin{equation}
  \Hol_I(\cF\mid M\cH)
  =\Hol_I(\cF M\mid\cH)
  \label{eq:holographic-identity}
\end{equation}
for every corresponding planar instance, where $\cF M$ denotes the family
of transforms in~\eqref{eq:left-transform}.  The cyclic order is unchanged.

\subsection{Realized supports, Gaussian mobility, and flags}
\label{sec:support-terminology}

Fix a qutrit presentation $(\cF,\cH,M)$ with $V=\C^3$ and a width-$t$ base
$M$ of row rank three.  Consider the all-left gadget closure of its induced
problem $\cF\mid M\cH$.  For a connected all-left gadget $\Gamma$ with
nonzero boundary tensor $T_\Gamma$, and a distinguished boundary edge $e$,
the subspace
\[
  S_e(\Gamma):=S_{j(e)}(T_\Gamma)\le V^*
\]
is a \emph{realized support}.  Thus ``realized'' always means that the
support occurs at an actual port of an actual connected planar gadget, not
merely in a linear combination of tensors.  Define
\begin{equation}
  \mathcal S_d
  :=\{S_e(\Gamma):\dim S_e(\Gamma)=d\},
  \qquad d\in\{1,2\}.
  \label{eq:support-closures}
\end{equation}
The \emph{support geometry} of the closure is the pair
$(\mathcal S_1,\mathcal S_2)$ together with the containment relations among
its realized rays and planes and, for each realized plane, the two parity
endpoints induced by $M$.
\emph{Cyclic rerooting} at a distinguished $t$-wire Boolean block means
rotating the boundary order until that block comes first and designating it
as the \emph{input} block; the wires inside the block retain their inherited
order.  To check the coordinates, let $s$ wires remain, write the physical
block flattening of the rerooted matchgate tensor $G$ as
\[
  F_G(z,w):=G(z,w),
  \qquad z\in\bits^t,\quad w\in\bits^s,
\]
and let
\[
  A_G(z,y):=G\bigl(z,\rho_s(y)\bigr)
\]
be its $t$-input, $s$-output exact matchgate matrix.  Thus $A_G=F_G R_s$
for the column-permutation matrix $R_s$ induced by $\rho_s$, and
\begin{equation}
  \operatorname{col}(F_G)
  =\operatorname{col}(A_G)=\row(A_G^{\trans}).
  \label{eq:rerooted-block-row-space}
\end{equation}
We call $\operatorname{col}(F_G)$ the \emph{block support} at the
distinguished block; it is the support of the whole block, rather than of
one Boolean wire.  By \cref{lem:ordered-clifford-realization},
$A_G^{\trans}$ is exact.  Hence
the reversal acts only on the complementary column index and never on the
distinguished block.  Consequently a realized line maps literally to a pure matchgate ray,
and a realized plane $P$ maps literally to the rank-two Gaussian row space
$PM$, with pullback parity endpoints $P^+,P^-$ as defined in
\cref{sec:matchgate-matrix-preliminaries}.

\begin{definition}[Monomial flag and connection-primality]
\label[definition]{def:monomial-flag}
Suppose $P\in\mathcal S_2$ and $L\in\mathcal S_1$ with
$L\not\subset P$.  The pair $(P,L)$ is a \emph{monomial flag} if
\[
  \mathcal S_2=\{P\},
  \qquad
  \mathcal S_1\subseteq\{P^+,P^-,L\}.
\]
The all-left gadget closure is \emph{connection-prime} if its realized
supports span $V^*$, $\mathcal S_2$ is nonempty, and it admits no monomial
flag.
\end{definition}

An \emph{adapted ray basis} for a flag $(P,L)$ is an ordered basis of
$V^*$ obtained by choosing nonzero generators of $P^+$, $P^-$, and $L$.
Changing those generators only rescales the three coordinates.

The names used in the trichotomy are now literal.  A \emph{ray-only} or
\emph{ray-separable} closure has $\mathcal S_2=\varnothing$.  A
\emph{Gaussian-mobile} closure has a realized plane but no monomial flag;
the proofs show that its moving planes or rays generate a constant-rank
Gaussian cover.  A \emph{one-plane flag closure} admits a monomial flag.
A \emph{Boolean matchgate core on $P$ times pure-ray factors} is a tensor
obtained by encoding an $m$-bit signature $h\in\MG_m$ in an ordered endpoint
basis $(p_0,p_1)$ with $p_0\in P^+$ and $p_1\in P^-$ at the plane-supported
ports, and tensoring nonzero generators of
$P^+$, $P^-$, or $L$ at the remaining ports.  The exact coordinate formula
is proved in \cref{lem:flag-normal-form}.  \emph{Rank-one stripping} is the
elementary step of factoring off such a ray at a line-supported port.
For a fixed adapted basis, the \emph{residual right tables} are the labelled
coordinate arrays obtained from induced right tensors after contracting
selected ports with fixed pure-ray factors and expressing all remaining
plane-supported ports in the dual endpoint coordinates.  Their port order
and labels are retained.  This phrase names the right-side data left
uncontrolled by the all-left flag normal form; it does not introduce a new
gadget operation or a new notion of equivalence.

\section{Matchgate Parameter Bounds and Star-Table Varieties}
\label{sec:matchgate-varieties}

The technical framework of
\cref{sec:matchgate-preliminaries,sec:algebraic-preliminaries}
turns the width question into a parameter-counting problem.  We begin the
negative lane of \cref{fig:conceptual-roadmap} by asking how much independent
information a width-$r$ labelled matchgate star can encode.  We first bound
the parameter dimension available to a single matchgate and then propagate
that bound through sampled block contractions.  The resulting
$\Q$-defined star-table closures are the low-width algebraic loci that
\cref{sec:obstruction-tables} will avoid.

For $s\ge0$, set
\[
  \delta(s):=1+\binom{s}{2}.
\]

\begin{lemma}[Coordinate transcendence bound]
\label[lemma]{lem:matchgate-transcendence}
For every integer $s\ge0$ and every $G\in\MG_s$,
\begin{equation}
  \trdeg_{\Q}\Q(G(x):x\in\bits^s)
  \le \delta(s).
  \label{eq:matchgate-transcendence-bound}
\end{equation}
\end{lemma}

\begin{proof}
For $s=0$, the signature consists of one scalar, so the claim is immediate.
The zero signature is also immediate.  Suppose therefore that $s\ge1$ and
$G\ne0$.  Choose a pivot $X\in\bits^s$ with $c:=G(X)\ne0$.

Matchgate parity implies that $G(Y)=0$ whenever $\dist(X,Y)$ is odd.  The
standard pivoted Pfaffian chart for the matchgate identities
\cite{CaiGorenstein2014} gives fixed signs $\eta_{ij}\in\{+1,-1\}$ such that
\begin{equation}
  a_{ij}
  =\eta_{ij}\frac{G(X\oplus\one_{\{i,j\}})}{c}
  \qquad(1\le i<j\le s)
  \label{eq:pivot-skew-entries}
\end{equation}
are the upper-triangular entries of a skew-symmetric matrix $A_X$.  For
every even subset $T\subseteq[s]$, there is another fixed sign
$\eta_T\in\{+1,-1\}$ for which
\begin{equation}
  G(X\oplus\one_T)
  =\eta_T c\,\Pf(A_X[T]).
  \label{eq:pivot-pfaffian-chart}
\end{equation}
Here $\one_T$ is the indicator bit string of $T$, and $\oplus$ is bitwise
exclusive-or.  The signs depend only on the pivot $X$, the indicated index
set, and the ordered-coordinate convention, not on the coordinate values of
$G$.

Equations~\eqref{eq:pivot-skew-entries} and
\eqref{eq:pivot-pfaffian-chart} show that every coordinate of $G$ lies in
\[
  \Q\bigl(
    c,\{G(X\oplus\one_{\{i,j\}}):1\le i<j\le s\}
  \bigr).
\]
This field has at most $1+\binom{s}{2}$ generators over $\Q$, which proves
\eqref{eq:matchgate-transcendence-bound}.  The argument uses a pivot of the
particular nonzero signature and treats both parity components; the zero
signature was handled separately.
\end{proof}

The next proposition packages the information bound in the form needed for
labelled closed-star instances.  We call $L$ the \emph{sampler alphabet}, the
family $(g_\ell)_{\ell\in L}$ the \emph{samplers}, and $T$ the corresponding
\emph{sampled star table}.  When $L=\bits$, $(g_0,g_1)$ is the
\emph{sampler pair}, and the displayed contraction is a \emph{star-table
representation}.

\begin{proposition}[Sampled star-table bound]
\label[proposition]{prop:star-table-bound}
Let $L$ be a finite set of size $q\ge1$, and let $k\ge1$ and $r\ge0$.
Suppose
\[
  g_\ell\in\MG_r\quad(\ell\in L),
  \qquad
  Q\in\MG_{kr}.
\]
Define a table $T\colon L^k\to\C$ by
\begin{equation}
  T(\ell_1,\ldots,\ell_k)
  =\Contr(Q;g_{\ell_1},\ldots,g_{\ell_k}),
  \label{eq:sampled-star-table}
\end{equation}
where the $j$th $r$-bit tensor factor is contracted with the $j$th
consecutive block of $Q$.  Then
\begin{equation}
  \trdeg_{\Q}\Q(T(\boldsymbol\ell):\boldsymbol\ell\in L^k)
  \le q\delta(r)+\delta(kr).
  \label{eq:sampled-star-bound}
\end{equation}
\end{proposition}

\begin{proof}
Let $K$ be the field generated over $\Q$ by all coordinates of
$\{g_\ell:\ell\in L\}$ and $Q$.  Every coordinate in
\eqref{eq:sampled-star-table} is a polynomial with integer coefficients in
these generators, and hence
\[
  \Q(T(\boldsymbol\ell):\boldsymbol\ell\in L^k)\subseteq K.
\]
Transcendence degree is subadditive under field composita.  Applying
\cref{lem:matchgate-transcendence} to the $q$ unary-block signatures and to
$Q$ gives
\[
  \trdeg_{\Q}K
  \le\sum_{i=1}^q \delta(r)+\delta(kr)
  =q\delta(r)+\delta(kr),
\]
as required.
\end{proof}

We now promote the same parameter count to an algebraic-geometric
statement.  This is the step that permits an integer obstruction even
though integer coordinates have transcendence degree zero.

For $s\ge1$ and $\nu\in\bits$, let
$\mathfrak M_s^\nu\subseteq\mathbb A_{\Q}^{2^s}$ be the affine
algebraic set obtained by imposing the parity-$\nu$ coordinate
equations $Y_x=0$ whenever $|x|\not\equiv\nu\pmod 2$, together with the
matchgate identities in the fixed external order.  All these equations
have integer coefficients.  Set
\[
  \mathfrak M_s:=\mathfrak M_s^0\cup\mathfrak M_s^1.
\]
This is a $\Q$-defined closed algebraic set; its defining ideal may be
taken to be the intersection of the two parity-component ideals.
Necessity and sufficiency of the matchgate identities, together with
parity, imply
\begin{equation}
  \mathfrak M_s(\C)=\MG_s,
  \label{eq:matchgate-variety-points}
\end{equation}
where the zero tensor belongs to both parity
components~\cite{CaiGorenstein2014}.  Separately, set
$\mathfrak M_0:=\mathbb A^1_{\Q}$, so that
$\mathfrak M_0(\C)=\MG_0=\C$.

\begin{lemma}[Dimension of the matchgate variety]
\label[lemma]{lem:matchgate-variety-dimension}
For every $s\ge0$,
\begin{equation}
  \dim \mathfrak M_s\le \delta(s)=1+\binom{s}{2}.
  \label{eq:matchgate-variety-dimension}
\end{equation}
\end{lemma}

\begin{proof}
The assertion is immediate for $s=0$.  For $s\ge1$, dimension is unchanged
under the field extension $\Q\subseteq\C$, so we work after base change to
$\C$.

Fix one parity component and a bit string $X$ of that parity.  Let $U_X$
be the principal open subset on which the coordinate indexed by $X$ is
nonzero.  For fixed pivot-dependent signs
$\eta_{X,T}\in\{+1,-1\}$, define a regular map
\[
  \psi_X\colon
  \mathbb G_m\times\mathbb A^{\binom{s}{2}}
  \longrightarrow
  \mathbb A^{2^s}
\]
by
\begin{equation}
  \psi_X(c,A)(X\oplus\one_T)
  =
  \begin{cases}
    \eta_{X,T}\,c\,\Pf(A[T]),&|T|\text{ even},\\
    0,&|T|\text{ odd},
  \end{cases}
  \qquad \eta_{X,\varnothing}=1,
  \label{eq:variety-pivot-chart}
\end{equation}
where $A$ ranges over the skew-symmetric $s\times s$ matrices.  This map is
defined over $\Q$.  The pivoted Pfaffian chart
\eqref{eq:pivot-skew-entries}--\eqref{eq:pivot-pfaffian-chart} shows that
every point of $U_X$ lies in the image of $\psi_X$.  Hence
\[
  \dim U_X
  \le
  \dim\left(\mathbb G_m\times
             \mathbb A^{\binom{s}{2}}\right)
  =1+\binom{s}{2}.
\]
The finitely many $U_X$ cover all nonzero points of the chosen parity
component; adjoining the zero point does not increase dimension.  Taking
the maximum over the two parity components proves the claim.
\end{proof}

Fix $k\ge1$ and $r\ge0$.  On
\begin{equation}
  \mathfrak X_{k,r}
  :=\mathfrak M_r\times\mathfrak M_r\times\mathfrak M_{kr},
  \label{eq:star-parameter-space}
\end{equation}
define the $\Q$-morphism
\[
  \Phi_{k,r}\colon\mathfrak X_{k,r}
  \longrightarrow\mathbb A_{\Q}^{2^k}.
\]
Index the target coordinates by bit strings
$\boldsymbol z=(z_1,\ldots,z_k)\in\bits^k$ and set
\begin{equation}
  [\Phi_{k,r}(g_0,g_1,Q)](\boldsymbol z)
  :=\Contr(Q;g_{z_1},\ldots,g_{z_k})
  \qquad(\boldsymbol z\in\bits^k).
  \label{eq:star-table-morphism}
\end{equation}
Every output coordinate is a polynomial with integer coefficients in the
source coordinates.  This formula also covers $r=0$: the three matchgate
signatures are then scalars and every block has the unique empty word.
Identifying $\mathbb A_{\Q}^{2^k}(\C)$ with the table space
$\C^{\bits^k}$, let
\begin{equation}
  \mathcal S_{k,r}:=\Phi_{k,r}(\mathfrak X_{k,r}(\C))
  \subseteq\C^{\bits^k},
  \qquad
  \mathfrak Z_{k,r}:=\overline{\mathcal S_{k,r}}^{\mathrm{Zar}}
  \subseteq\C^{\bits^k}.
  \label{eq:star-table-variety}
\end{equation}

\begin{proposition}[Algebraic star-table bound]
\label[proposition]{prop:algebraic-star-table-bound}
The variety $\mathfrak Z_{k,r}$ is defined over $\Q$ and satisfies
\begin{equation}
  \dim\mathfrak Z_{k,r}
  \le D_{k,r}:=2\delta(r)+\delta(kr).
  \label{eq:algebraic-star-table-dimension}
\end{equation}
In particular, if $D_{k,r}<2^k$, then $\mathfrak Z_{k,r}$ is a proper
$\Q$-defined Zariski-closed subset of $\C^{\bits^k}$.
\end{proposition}

\begin{proof}
By \cref{lem:matchgate-variety-dimension},
\[
  \dim\mathfrak X_{k,r}
  \le2\delta(r)+\delta(kr).
\]
The dimension of the Zariski closure of the image of a morphism is at most
the dimension of its source, proving
\eqref{eq:algebraic-star-table-dimension}.

It remains to verify that $\mathfrak Z_{k,r}$ is defined over $\Q$.
Let $\mathbf U$ denote all source-coordinate variables and let
$\mathbf Y=(Y_{\boldsymbol z})_{\boldsymbol z\in\bits^k}$ be the target
coordinates.  Choose a
$\Q$-coefficient ideal
\[
  I_{k,r}\subseteq\Q[\mathbf U]
\]
whose complex zero set is $\mathfrak X_{k,r}(\C)$.  In
$\Q[\mathbf U,\mathbf Y]$, form the graph ideal
\[
  J_{k,r}
  :=I_{k,r}
  +\left\langle
    Y_{\boldsymbol z}-[\Phi_{k,r}(\mathbf U)](\boldsymbol z)
    :\boldsymbol z\in\bits^k
  \right\rangle
\]
and its elimination ideal
\[
  E_{k,r}:=J_{k,r}\cap\Q[\mathbf Y].
\]
After extending scalars from $\Q$ to $\C$, the elimination theorem gives
\[
  \mathfrak Z_{k,r}=V_{\C}(E_{k,r}).
\]
Equivalently, an elimination-order Gr\"obner basis computed over $\Q$
remains a Gr\"obner basis after extension to $\C$.  Thus
$\mathfrak Z_{k,r}$ is defined over $\Q$.  If $D_{k,r}<2^k$, its dimension
is smaller than that of the ambient affine space, so it is proper.
\end{proof}

\section{Two Obstruction Tables}
\label{sec:obstruction-tables}

The parameter budget of \cref{sec:matchgate-varieties} reduces the next step
to an avoidance problem: find a $2^k$-entry table outside every relevant
low-width star-table closure.  Using the tensor and algebraic-geometric
conventions of
\cref{sec:tensor-preliminaries,sec:algebraic-preliminaries}, we obtain two
nowhere-zero witnesses whose one-versus-rest flattenings all have rank two:
a closed-form transcendental table and a positive-integer table selected by
finite Zariski avoidance.  The integer witness drives the rational qutrit
construction, while the transcendental witness gives an explicit companion.
The next section addresses the complementary question of realizing these
hard tables exactly at a larger finite width.

\subsection{A closed-form transcendental table}
\label{sec:independent-table}

We next construct a table with the maximum possible coordinate
transcendence degree $2^k$.
The only transcendence-theoretic input is the following standard consequence
of the Lindemann--Weierstrass theorem.

\begin{lemma}
\label[lemma]{lem:independent-exponentials}
Let $p_1,\ldots,p_N$ be distinct primes.  Then
\[
  \exp(\sqrt{p_1}),\ldots,\exp(\sqrt{p_N})
\]
are algebraically independent over $\Q$.
\end{lemma}

\begin{proof}
The prime square classes are independent in
$\Q^{\times}/(\Q^{\times})^2$: unique factorization shows that a product
$\prod_i p_i^{\epsilon_i}$ with $\epsilon_i\in\{0,1\}$ is a square only
when every $\epsilon_i=0$.  Hence the multiquadratic field
$\Q(\sqrt{p_1},\ldots,\sqrt{p_N})$ has the automorphisms that independently
change the sign of any one square root; an automorphism here is a bijection
preserving addition, multiplication, and every rational number.  Applying
such an automorphism to a
putative rational linear relation and subtracting shows that
$\sqrt{p_1},\ldots,\sqrt{p_N}$ are linearly independent over $\Q$.

Suppose that a nonzero polynomial $P\in\Q[X_1,\ldots,X_N]$ vanished at the
displayed exponentials.  Expanding its monomials would give
\begin{equation}
  \sum_{\mathbf m\in\mathcal S}
    c_{\mathbf m}
    \exp\!\left(\sum_{i=1}^N m_i\sqrt{p_i}\right)=0,
  \label{eq:expanded-polynomial-relation}
\end{equation}
for a finite nonempty $\mathcal S\subseteq\mathbb Z_{\ge0}^N$ and nonzero
rational coefficients $c_{\mathbf m}$.  Linear independence of the square
roots makes the algebraic exponents in
\eqref{eq:expanded-polynomial-relation} pairwise distinct, including the
zero exponent if $P$ has a constant monomial.  The
Lindemann--Weierstrass theorem says that exponentials of distinct algebraic
numbers are linearly independent over $\overline{\Q}$
\cite{Popescu2024}.  This contradicts
\eqref{eq:expanded-polynomial-relation}.
\end{proof}

Fix $k\ge2$, and let $p_1,\ldots,p_{2^k}$ be the first $2^k$ primes.  For
$z=(z_1,\ldots,z_k)\in\bits^k$, define its one-based lexicographic index by
\begin{equation}
  \ind_k(z)
  :=1+\sum_{j=1}^k z_j2^{k-j}
  \in[2^k],
  \label{eq:lexicographic-index}
\end{equation}
and set
\begin{equation}
  f_k^{\mathrm{tr}}(z):=\exp\!\left(\sqrt{p_{\ind_k(z)}}\right).
  \label{eq:independent-table-definition}
\end{equation}
Every entry is positive and nonzero.  By
\cref{lem:independent-exponentials},
\begin{equation}
  \trdeg_{\Q}\Q(f_k^{\mathrm{tr}}(z):z\in\bits^k)=2^k.
  \label{eq:independent-table-transcendence}
\end{equation}

\begin{proposition}[Explicit transcendental obstruction]
\label[proposition]{prop:transcendental-obstruction-table}
Every one-versus-rest flattening of $f_k^{\mathrm{tr}}$ has rank two.
Moreover, if $r\ge0$, $g_0,g_1\in\MG_r$, and $Q\in\MG_{kr}$ satisfy
\[
  f_k^{\mathrm{tr}}(z)
  =\Contr(Q;g_{z_1},\ldots,g_{z_k})
  \qquad(z\in\bits^k),
\]
then
\begin{equation}
  2^k\le2\delta(r)+\delta(kr).
  \label{eq:transcendental-obstruction-bound}
\end{equation}
\end{proposition}

\begin{proof}
For each port, choose two distinct assignments to the other $k-1$
coordinates.  The corresponding $2$-by-$2$ determinant is the difference
of two distinct degree-two monomials in four different entries of
$f_k^{\mathrm{tr}}$.  It is nonzero because those entries are algebraically
independent by \eqref{eq:independent-table-transcendence}.  Thus every
mode flattening has rank two.

Apply \cref{prop:star-table-bound} with the sampler alphabet
$L=\{0,1\}$, then use
\eqref{eq:independent-table-transcendence}.
\end{proof}

\subsection{A positive-integer table}
\label{sec:integer-table}

We use the following elementary density fact.

\begin{lemma}[Density of the positive integer grid]
\label[lemma]{lem:positive-grid-density}
For every $N\ge1$, the set $\mathbb Z_{>0}^N$ is Zariski dense in
$\C^N$.
\end{lemma}

\begin{proof}
We induct on $N$.  A univariate polynomial vanishing at every positive
integer is zero.  For $N>1$, write a polynomial vanishing on
$\mathbb Z_{>0}^N$ as a polynomial in its last variable.  After fixing any
positive integer values of the first $N-1$ variables, the resulting
univariate polynomial has infinitely many roots, so every coefficient
vanishes on $\mathbb Z_{>0}^{N-1}$.  The induction hypothesis makes all
coefficient polynomials zero.
\end{proof}

\begin{proposition}[Positive-integer obstruction]
\label[proposition]{prop:integer-obstruction-table}
For every $k\ge2$, there exists a table
\begin{equation}
  f_k^{\mathrm{int}}\colon\bits^k\longrightarrow\mathbb Z_{>0}
  \label{eq:integer-table-definition}
\end{equation}
such that, for every $r\ge0$, $g_0,g_1\in\MG_r$, and
$Q\in\MG_{kr}$, the identities
\begin{equation}
  f_k^{\mathrm{int}}(z)
  =\Contr(Q;g_{z_1},\ldots,g_{z_k})
  \qquad(z\in\bits^k)
  \label{eq:integer-star-representation}
\end{equation}
imply
\begin{equation}
  2^k\le2\delta(r)+\delta(kr).
  \label{eq:integer-obstruction-bound}
\end{equation}
The conclusion allows the representing matchgates to have arbitrary
complex weights.  In addition, every one-versus-rest flattening of
$f_k^{\mathrm{int}}$ has rank two.
\end{proposition}

\begin{proof}
Let
\[
  R_k:=\{r\in\mathbb Z_{\ge0}:D_{k,r}<2^k\}.
\]
This set is finite because $D_{k,r}=2\delta(r)+\delta(kr)$ tends to infinity with
$r$.  Moreover, $D_{k,0}=3<2^k$ for $k\ge2$, so $0\in R_k$.  For each
$r\in R_k$, \cref{prop:algebraic-star-table-bound} makes
$\mathfrak Z_{k,r}$ a proper $\Q$-defined closed subset of
$\C^{\bits^k}$.  Its defining ideal therefore contains a nonzero
polynomial
$P_{k,r}\in\Q[(Y_{\boldsymbol z})_{\boldsymbol z\in\bits^k}]$.  The finite product
\[
  P_k:=\prod_{r\in R_k}P_{k,r}
\]
is nonzero.  For every port $j$, choose two distinct assignments
$u_j,v_j\in\bits^{k-1}$ and let $\Delta_j$ be the corresponding symbolic
$2$-by-$2$ flattening minor.  Each $\Delta_j$ is a nonzero polynomial, so
\[
  P_k\prod_{j=1}^k\Delta_j
\]
is nonzero.  By \cref{lem:positive-grid-density}, choose
$a=(a_{\boldsymbol z})_{\boldsymbol z\in\bits^k}
\in\mathbb Z_{>0}^{2^k}$ on which this product does not vanish, and set
$f_k^{\mathrm{int}}(\boldsymbol z)=a_{\boldsymbol z}$.  The nonzero minor
$\Delta_j(a)$ gives rank at least two in mode $j$; the two-row shape gives
rank at most two.  Thus every mode rank is exactly two.

Suppose \eqref{eq:integer-star-representation} holds at width $r$.  Then
$f_k^{\mathrm{int}}$ lies in the exact image $\mathcal S_{k,r}$ and hence
in $\mathfrak Z_{k,r}$.  Because $\mathcal S_{k,r}$ was defined using all
complex points of the matchgate parameter space, this implication places
no algebraicity or field-of-definition restriction on the representing
matchgates.  If $D_{k,r}<2^k$, then $r\in R_k$, which would
force $P_{k,r}(f_k^{\mathrm{int}})=0$, contrary to the choice of $a$.
Therefore $D_{k,r}\ge2^k$, proving
\eqref{eq:integer-obstruction-bound}.
\end{proof}

\begin{remark}
The integer table is existential but effective in principle: for each of
the finitely many relevant widths one may compute elimination equations,
then enumerate positive integer tables until their product does not vanish.
The proof supplies no useful running-time, height, or bit-length bound.  By
contrast, $f_k^{\mathrm{tr}}$ has a simple closed formula but uses
transcendental constants.  Keeping both versions separates these two
tradeoffs cleanly.
\end{remark}

\section{Pfaffian Interpolation on Two Block Codewords}
\label{sec:pfaffian-interpolation}

The tables of \cref{sec:obstruction-tables} force every matchgate-star
representation to obey the width inequality derived in
\cref{sec:matchgate-varieties}, but a counterexample also requires an exact
presentation at some finite width.  Using the Pfaffian and ordered-block
conventions of
\cref{sec:matchgate-preliminaries,sec:base-presentation}, we solve this
realization problem in a field-preserving form: every nowhere-zero Boolean
table can be interpolated by one exact matchgate on the two repetition
codewords in each block.  This construction supplies the Boolean preimage
that will be packaged into a qutrit presentation in the next section.

\begin{theorem}[Block-Pfaffian interpolation]
\label{thm:block-pfaffian-interpolation}
Let $K\subseteq\C$ be a subfield, let $k\ge1$, let
$f\colon\bits^k\to K^{\times}$, and put $t=2^k$.
There is an exact matchgate signature $H_f\in\MG_{kt}$ with coordinates in
$K$, on $k$ consecutive $t$-bit blocks, such that, for
\[
  c_0:=0^t,
  \qquad
  c_1:=1^t,
\]
one has
\begin{equation}
  H_f(c_{z_1},\ldots,c_{z_k})=f(z_1,\ldots,z_k)
  \qquad\text{for every }z\in\bits^k.
  \label{eq:block-interpolation-target}
\end{equation}
More precisely, the construction produces a named skew-symmetric matrix
$A_f\in K^{kt\times kt}$ whose principal-Pfaffian signature is normalized
by
\begin{equation}
  \Gamma_{A_f}(c_{z_1},\ldots,c_{z_k})
  =\frac{f(z)}{f(0^k)},
  \qquad
  \Gamma_{A_f}(c_0,\ldots,c_0)=1,
  \label{eq:block-interpolation-normalized-statement}
\end{equation}
and $H_f=f(0^k)\Gamma_{A_f}$.
The blocks and all wires within them occur in the specified external linear
order.  Moreover, $H_f$ has an exact planar matchgate realization whose
edge weights all lie in $K$.
\end{theorem}

The construction has one independent Pfaffian component for each nonempty
$S\subseteq[k]$.  Its principal Pfaffian will equal a multiplier $y_S$ when
every block in $S$ is selected and will equal $1$ otherwise.  Disjointness
then multiplies these contributions and implements a multiplicative
M\"obius expansion of $f$.

\begin{proof}
We first decompose the table multiplicatively.  For each nonempty subset
$S\subseteq[k]$, define $y_S\in K^{\times}$ recursively under inclusion by
\begin{equation}
  y_S
  :=
  \frac{f(\one_S)/f(0^k)}
       {\displaystyle\prod_{\varnothing\ne U\subsetneq S}y_U},
  \label{eq:multiplicative-mobius}
\end{equation}
where $\one_S\in\bits^k$ is the indicator of $S$.  The denominator is
nonzero.  The recursion determines the multipliers in increasing order of
$|S|$.  Induction on $|T|$ gives
\begin{equation}
  \prod_{\varnothing\ne S\subseteq T}y_S
  =\frac{f(\one_T)}{f(0^k)}
  \qquad(T\subseteq[k]),
  \label{eq:mobius-product-identity}
\end{equation}
with the empty product interpreted as $1$.

We now define a skew-symmetric matrix whose principal Pfaffians implement
the left side of~\eqref{eq:mobius-product-identity}.  In logical block
$i\in[k]$, create an adjacent ordered pair of modes
\[
  (a_{S,i},b_{S,i})
\]
for every subset $S\subseteq[k]$ containing $i$.  A fixed block contains
$2\cdot2^{k-1}=2^k=t$ modes.  Order the blocks by increasing $i$; within
each block, order the pairs lexicographically by $S$, with $a_{S,i}$
immediately before $b_{S,i}$; explicitly, subsets are compared by their
indicator words $\one_S$ in lexicographic order.

Let $A_f$ be the skew-symmetric matrix indexed by these $kt$ ordered modes.
All its entries lie in $K$.  All unspecified upper-triangular entries are
zero, and different subsets
$S$ have no entries between their modes.  For a singleton $S=\{i\}$, set
\begin{equation}
  A_f[a_{S,i},b_{S,i}]=y_S.
  \label{eq:singleton-component}
\end{equation}
For $S=\{i_1<\cdots<i_s\}$ with $s\ge2$, set the baseline entries
\begin{equation}
  A_f[a_{S,i_q},b_{S,i_q}]=1
  \qquad(1\le q\le s),
  \label{eq:baseline-entries}
\end{equation}
the consecutive alternate entries
\begin{equation}
  A_f[b_{S,i_q},a_{S,i_{q+1}}]=1
  \qquad(1\le q<s),
  \label{eq:alternate-path-entries}
\end{equation}
and the closing entry
\begin{equation}
  A_f[a_{S,i_1},b_{S,i_s}]=y_S-1.
  \label{eq:alternate-closing-entry}
\end{equation}
The entries in the reverse direction are fixed by skew-symmetry.

Fix one subset $S$ and suppose a set $T\subseteq[k]$ of logical blocks is
selected in full.  The selected modes in the $S$-component are precisely
the pairs indexed by $S\cap T$.  If $|S|=1$, then
\eqref{eq:singleton-component} immediately gives contribution $y_S$ when
$S\subseteq T$ and contribution $1$ otherwise.  Assume henceforth that
$|S|\ge2$.  The baseline and alternate positions in
\eqref{eq:baseline-entries}--\eqref{eq:alternate-closing-entry} form one
formal alternating cycle on the $2|S|$ modes.  The Pfaffian expansion may
be viewed on this cycle; if $y_S-1=0$, the alternate matching term is
simply zero.  If $S\nsubseteq T$,
deleting the unselected two-vertex pairs breaks this cycle into disjoint
even paths.  Each such path has its baseline edges as its unique nonzero
perfect-matching monomial, so the component Pfaffian is $1$; the same holds
when no pair is selected because $\Pf(A_f[\varnothing])=1$.

If $S\subseteq T$, all modes of the component remain.  The cycle has
exactly two formal perfect-matching terms, one possibly zero.  In the local order
\[
  a_{S,i_1},b_{S,i_1},\ldots,a_{S,i_s},b_{S,i_s},
\]
the baseline matching consists of adjacent pairs and has Pfaffian sign
$+1$.  The alternate matching consists of the pair in positions $(1,2s)$
and the adjacent pairs $(2,3),(4,5),\ldots,(2s-2,2s-1)$, so it is
noncrossing---no two paired intervals interleave in the displayed linear
order---and also has sign $+1$.  Its weight is $y_S-1$.  Thus the full
$S$-component contributes $1+(y_S-1)=y_S$.  This reasoning also allows
$y_S=1$, in which case the second term is zero; a zero-weight edge may be
omitted from the eventual planar matchgate.

Let $\Gamma_{A_f}$ be the principal-Pfaffian signature from
\eqref{eq:principal-pfaffian-tensor}.  Its coordinates are Pfaffian
polynomials with integer coefficients in the entries of $A_f$, so they lie
in $K$.  The principal-Pfaffian relations imply the matchgate identities;
by \cref{lem:field-preserving-realization}, $\Gamma_{A_f}$ has an exact
planar realization over $K$ in the stated order.  For a logical assignment
$z\in\bits^k$, let $T=\{i:z_i=1\}$ and evaluate $\Gamma_{A_f}$ on the block
word $(c_{z_1},\ldots,c_{z_k})$.  Regrouping the selected modes by $S$
permutes only adjacent two-coordinate chunks $(a_{S,i},b_{S,i})$.  Every
transposition of two chunks has sign $(-1)^{2\cdot2}=+1$, and regrouping
preserves increasing $i$-order within a fixed $S$.  Pfaffian factorization
over the disjoint components therefore introduces no global sign.  The
component analysis and~\eqref{eq:mobius-product-identity} give
\begin{equation}
  \Gamma_{A_f}(c_{z_1},\ldots,c_{z_k})
  =\prod_{\varnothing\ne S\subseteq T}y_S
  =\frac{f(z)}{f(0^k)}.
  \label{eq:pfaffian-interpolation-normalized}
\end{equation}

Finally, multiply the whole matchgate signature by the nonzero scalar
$f(0^k)\in K$, for example by taking a disjoint internal edge of that
weight.  The resulting exact matchgate signature $H_f$ is realized over
$K$ and satisfies
\eqref{eq:block-interpolation-target}.  The matrix order was block-major
throughout, so this construction also proves the asserted external linear
order.
\end{proof}

\begin{remark}
The construction specifies the target signature exactly through
\eqref{eq:principal-pfaffian-tensor},
\eqref{eq:multiplicative-mobius},
\eqref{eq:singleton-component}--\eqref{eq:alternate-closing-entry}, and the
final scalar multiplication by $f(0^k)$.
The field-preserving matchgate-identity sufficiency theorem supplies an
exact planar realization over the same coefficient field.  It is applied
to the Pfaffian tensor $\Gamma_{A_f}$; the parameter matrix $A_f$ itself is
not asserted to be planar.  No claim of polynomial-size encoding or
polynomial-time construction is intended: the theorem concerns exact
representability.
\end{remark}

\section{The Controlled Domain-Three Presentation}
\label{sec:domain-three-construction}

The interpolation theorem of \cref{sec:pfaffian-interpolation} realizes the
positive-integer obstruction from \cref{sec:obstruction-tables} on Boolean
blocks, but it does not yet produce a qutrit language with the structural
properties required by \cref{thm:main}.  Using the tensor-transform,
exact-matchgate, and common-base conventions of
\cref{sec:tensor-preliminaries,sec:matchgate-preliminaries,sec:base-presentation},
we add marker and control modes and a full-row-rank qutrit base, while
retaining the hard table as an exact restriction of the resulting right
tensor.  We then verify the local ranks, two-sided support-essentiality, and
connected two-center coupling.  This completes the construction side; the
next section asks whether an arbitrary equivalent presentation can compress
it.

Let $f=f_k^{\mathrm{int}}$ be the positive-integer obstruction table from
\cref{prop:integer-obstruction-table}, put
\[
  t_0:=2^k,
  \qquad T_k:=t_0+3,
\]
and use the normalized skew matrix $A_f$ named in
\cref{thm:block-pfaffian-interpolation}; normalized means precisely
\eqref{eq:block-interpolation-normalized-statement}.  Put
$c_0:=0^{t_0}$ and $c_1:=1^{t_0}$.  Thus, for $z\in\bits^k$,
$\Gamma_{A_f}(c_{z_1},\ldots,c_{z_k})=f(z)/f(0^k)$.

One qutrit block consists of $T_k$ Boolean wires in the order
\[
  (\text{marker};\ \text{$t_0$ selectors};\ \text{two controls}).
\]
Define the three block codewords
\begin{equation}
  b_0=(0,0^{t_0},00),\qquad
  b_1=(1,0^{t_0},00),\qquad
  b_2=(0,1^{t_0},11).
  \label{eq:controlled-codewords}
\end{equation}
Since $t_0$ is even, $b_0,b_2$ have even parity and $b_1$ has odd parity.
With $e_x$ denoting a Boolean coordinate vector, set
\begin{equation}
  M_k=
  \begin{pmatrix}
    e_{b_0}^{\trans}\\ e_{b_1}^{\trans}\\ e_{b_2}^{\trans}
  \end{pmatrix}
  \in\Q^{3\times2^{T_k}}.
  \label{eq:Mk-definition}
\end{equation}
The rows are distinct coordinate matchgates, so $M_k$ has row rank three.

The left family contains the pins $\delta_0,\delta_1,\delta_2$.  For
$a,b\in D_3$, let the qutrit matrix unit $E_{ab}$ be defined by
$E_{ab}(u,v)=1$ when $(u,v)=(a,b)$ and $0$ otherwise, and define the binary
signature
\[
  X_{01}:=E_{01}+E_{10}.
\]
Their transforms are exact matchgates: the pins become $e_{b_a}^{\trans}$,
and
\begin{equation}
  X_{01}M_k^{\otimes2}
  =e_{b_0,b_1}^{\trans}+e_{b_1,b_0}^{\trans},
  \label{eq:X01-transform}
\end{equation}
which is the odd two-bit disequality matchgate on the two marker modes,
tensored with fixed zero pins.

It remains to construct the right Boolean preimage.  It has $k+2$
consecutive blocks: two \emph{link blocks}, followed by $k$ \emph{hard
blocks}.  Fix one designated hard block $j\in[k]$, and set
\[
  \alpha=1,\qquad \beta=2,
  \qquad w(0)=\alpha,\quad w(2)=\beta.
\]
Construct a skew matrix $\widehat A_f$ on these $(k+2)T_k$ modes as
follows.
\begin{enumerate}
  \item On the selector modes of the hard blocks, use $A_f$.
  \item Give zero rows and columns to every hard marker and to every
  nonmarker mode of the two link blocks.
  \item On the control pair of every nondesignated hard block, put one
  isolated $2$-by-$2$ skew block of weight $1$: the entry from the first
  control to the second is $1$, its reverse is $-1$, and all other entries
  incident with that pair are zero.
  \item Let $p,q$ be the link markers and $a,b$ the controls of the
  designated hard block.  In the external order $p<q<a<b$, put
  \begin{equation}
    \widehat A_f[a,b]=1,\quad
    \widehat A_f[p,q]=\alpha,\quad
    \widehat A_f[p,a]=1,\quad
    \widehat A_f[q,b]=\alpha-\beta,
    \label{eq:control-block-entries}
  \end{equation}
  with every other entry involving these four modes equal to zero.
\end{enumerate}
All other entries between the listed components are zero.  Let
$\widehat H_k$ be $f(0^k)$ times the principal-Pfaffian signature of
$\widehat A_f$, and define
\begin{equation}
  \cF_k:=\{\delta_0,\delta_1,\delta_2,X_{01}\},
  \qquad \cH_k:=\{\widehat H_k\},
  \qquad B_k:=M_k^{\otimes(k+2)}\widehat H_k.
  \label{eq:controlled-language}
\end{equation}

\begin{proposition}[Exact restriction and rational presentation]
\label[proposition]{prop:valid-domain-three-presentation}
The triple $(\cF_k,\cH_k,M_k)$ is a valid width-$T_k$ exact matchgate
presentation over $\Q$, with a full-row-rank base.  Its right tensor is
integer-valued.  The only nonzero domain coordinates of $B_k$ have link
states in $\{0,1\}$ and hard states in $\{0,2\}$.  For $x\in\bits^k$,
writing $2x=(2x_1,\ldots,2x_k)$, one has
\begin{align}
  B_k(0,0,2x)&=f(x),
  &B_k(1,1,2x)&=w(2x_j)f(x),
  \label{eq:Bk-coordinate-table}\\
  B_k(0,1,2x)&=0,
  &B_k(1,0,2x)&=0.
  \notag
\end{align}
\end{proposition}

\begin{proof}
The left transforms were verified in \eqref{eq:X01-transform}.  The scaled
principal-Pfaffian signature $\widehat H_k$ has rational coordinates and,
by \cref{lem:field-preserving-realization}, an exact planar realization
with rational weights in the prescribed order.  Hence the presentation is
valid over $\Q$.

A principal-Pfaffian coordinate indexed by a codeword selects exactly the
matrix rows and columns at its $1$-bits.  Hence a link state $2$ selects
isolated link selectors and controls, while a hard
state $1$ selects an isolated hard marker; either choice gives Pfaffian
zero.  The four-mode control block satisfies
\[
  \Pf(\widehat A_f[\varnothing])=1,\quad
  \Pf(\widehat A_f[\{a,b\}])=1,\quad
  \Pf(\widehat A_f[\{p,q\}])=\alpha,\quad
  \Pf(\widehat A_f[\{p,q,a,b\}])
  =\alpha-(\alpha-\beta)=\beta.
\]
Selecting exactly one link marker has odd parity and gives zero.  The hard
selector blocks contribute $f(x)/f(0^k)$ by
\eqref{eq:pfaffian-interpolation-normalized}; the global scale restores
$f(x)$.
In each surviving case, all added selected pieces have even size, so
regrouping introduces no
Pfaffian shuffle sign.  This proves \eqref{eq:Bk-coordinate-table} and the
coordinate-support restriction.  Since $f$ and $w$ are positive integers, $B_k$ is
integer-valued.
\end{proof}

\begin{proposition}[Deficiency and essentiality]
\label[proposition]{prop:controlled-ranks}
Every port of $X_{01}$ and $B_k$ has flattening rank exactly two, while
the pins have rank one.  The language is support-essential on both sides,
and $B_k$ alone is support-essential on the right.  Its link-port support
is $\Span(e_0,e_1)$ and its hard-port support is $\Span(e_0,e_2)$.
\end{proposition}

\begin{proof}
The assertion for $X_{01}$ is immediate.  At either link port of $B_k$,
the state-$0$ and state-$1$ rows have disjoint nonzero supports, while the
state-$2$ row vanishes.  Thus the link support is exactly
$\Span(e_0,e_1)$.

At every hard port, the state-$1$ row vanishes.  At a nondesignated port,
the chosen nonzero $2$-by-$2$ minor of the corresponding flattening
$\Mat_j(f)$ makes the state-$0$ and state-$2$
rows independent.  At the designated port, proportionality of those rows
by a scalar $\lambda\in\C$ would, after fixing the other hard coordinates
to some $x'\in\bits^{k-1}$, give the following.  Let $x^{(a)}\in\bits^k$
be the extension of $x'$ having value $a$ at coordinate $j$.  The link
slice $(0,0)$ gives
$f(x^{(1)})=\lambda f(x^{(0)})$, while the slice $(1,1)$ gives
$\beta f(x^{(1)})=\lambda\alpha f(x^{(0)})$.  Since $f$ is nowhere zero, this
would force $\alpha=\beta$, a contradiction.  Hence every hard support is
exactly $\Span(e_0,e_2)$.

These two planes span the right qutrit space.  On the left, the three pins
span the dual space.  This proves support-essentiality and every rank
claim.
\end{proof}

For $z\in\{0,2\}^k$, define the relabelled table
\begin{equation}
  \widetilde f(z):=f(z/2),
  \qquad z/2:=(z_1/2,\ldots,z_k/2)\in\bits^k.
  \label{eq:relabelled-hard-table}
\end{equation}

\begin{proposition}[A genuinely coupled two-center gadget]
\label[proposition]{prop:two-center-coupling}
Write the two link ports of each copy, in their prescribed local cyclic
order, as $\ell_1,\ell_2$, followed by the hard ports
$h_1,\ldots,h_k$.  Join $\ell_2^L$ to $\ell_1^R$ through one copy of
$X_{01}$, and join $\ell_1^L$ to $\ell_2^R$ through another copy.  Draw
the first length-two path above the two centres and the second below them.
Thus the link \emph{names} are cross-paired, but the two drawn paths do
not cross.  At each centre the interior corner of the resulting
four-cycle is the empty cyclic sector from $\ell_1$ to $\ell_2$; the
complementary sector contains $h_1,\ldots,h_k$ and lies in the outer face.
Both copies therefore retain their prescribed local port orders.  Marking
$h_1^L$ first, the gadget boundary order can be taken to be
\[
  h_1^L,\ldots,h_k^L,h_1^R,\ldots,h_k^R.
\]
On hard assignments
$z,y\in\{0,2\}^k$, its boundary tensor is
\begin{equation}
  K_f(z,y)=\widetilde f(z)\widetilde f(y)
             \bigl(w(z_j)+w(y_j)\bigr),
  \label{eq:two-center-value}
\end{equation}
Its Schmidt rank across $z\mid y$ is exactly two.
\end{proposition}

\begin{proof}
Let $a$ and $b$ be the common link states at the left and right centres,
respectively.  Cross-pairing the link names does not permute the arguments
at either centre.  Each $B_k$ forces its two link states to agree, and the
symmetry of $X_{01}$ means that both copies impose the same condition
$a\ne b$.  Exactly the two pairs $(a,b)=(0,1)$ and $(1,0)$ survive, giving
\eqref{eq:two-center-value}.  The full boundary matrix is the sum of two
outer products
\[
  [\widetilde f(z)w(z_j)]_z[\widetilde f(y)]_y^{\trans}
  +[\widetilde f(z)]_z[\widetilde f(y)w(y_j)]_y^{\trans},
\]
and every entry involving hard state $1$ vanishes.  Its rank is therefore
at most two.  Multiplication on the left and right by the invertible
diagonal matrices with entries $\widetilde f(z)^{-1}$ and
$\widetilde f(y)^{-1}$ preserves rank.  After these scalings and
restricting to the two values of $z_j,y_j$, the remaining matrix is
\[
  \begin{pmatrix}
    2\alpha&\alpha+\beta\\
    \alpha+\beta&2\beta
  \end{pmatrix},
\]
whose determinant is $-(\alpha-\beta)^2\ne0$.  The Schmidt rank is exactly
two, so the gadget cannot factor as $U(z)V(y)$ into one tensor for each
centre.
\end{proof}

For comparison, retain the closed-form table as a companion construction.
Let
\begin{equation}
  K_k^{\mathrm{tr}}
  :=\Q\bigl(f_k^{\mathrm{tr}}(z):z\in\bits^k\bigr).
  \label{eq:transcendental-companion-definition}
\end{equation}
Repeat the controlled construction
\eqref{eq:controlled-codewords}--\eqref{eq:controlled-language} with
$f_k^{\mathrm{tr}}$ in place of $f_k^{\mathrm{int}}$, using the same
base $M_k$, left family $\cF_k$, and control weights $\alpha=1,\beta=2$.
Denote the resulting Boolean preimage and right tensor by
$\widehat H_k^{\mathrm{tr}}$ and $B_k^{\mathrm{tr}}$.

\begin{proposition}[Closed-form companion presentation]
\label[proposition]{prop:transcendental-companion-presentation}
The triple
$(\cF_k,\{\widehat H_k^{\mathrm{tr}}\},M_k)$ is a valid common
width-$T_k$ presentation over $K_k^{\mathrm{tr}}$.  The coordinate
restriction \eqref{eq:Bk-coordinate-table}, every rank and essentiality
conclusion of \cref{prop:controlled-ranks}, and the two-center conclusion
of \cref{prop:two-center-coupling} all hold with
$f_k^{\mathrm{tr}},B_k^{\mathrm{tr}}$ in place of $f,B_k$.
\end{proposition}

\begin{proof}
Apply \cref{thm:block-pfaffian-interpolation} over
$K_k^{\mathrm{tr}}$ and repeat the proofs of
\cref{prop:valid-domain-three-presentation,prop:controlled-ranks,prop:two-center-coupling}.
Those arguments use only nonvanishing and selected nonzero $2$-by-$2$ minors
of the one-versus-rest flattenings of the table, together with
$\alpha\ne\beta$.  The explicit table has those
properties by \cref{prop:transcendental-obstruction-table}.
\end{proof}

Recall the definition of $\mu_k$ from \eqref{eq:muk-definition}.
The presentation in \cref{prop:valid-domain-three-presentation} shows
\begin{equation}
  \mu_k\le T_k=2^k+3.
  \label{eq:muk-upper-bound}
\end{equation}
In particular, $\mu_k$ is a well-defined finite integer.

For the companion problem, put
\begin{equation}
  \mu_k^{\mathrm{tr}}
  :=\bw(\cF_k\mid\{B_k^{\mathrm{tr}}\}).
  \label{eq:transcendental-width-definition}
\end{equation}
By \cref{prop:transcendental-companion-presentation},
$\mu_k^{\mathrm{tr}}\le2^k+3$.

\section{The Lower Bound against Arbitrary Equivalent Presentations}
\label{sec:arbitrary-equivalence}

The preceding section supplies a valid width-$(2^k+3)$ presentation, and
therefore only an upper bound on its minimum exact common width.  A competing
exactly equivalent presentation may change its finite domain, common base,
coordinate tensors, and complex weights.  Using the star-contraction
conventions of \cref{sec:base-presentation} and the equivalence notion of
\cref{sec:equivalence-definition}, we exploit the preservation of closed
labelled planar instances: a fixed family of labelled stars recovers the
same obstruction table from every competing presentation.  The obstruction
of \cref{sec:obstruction-tables} then yields the quantitative lower bound and
\cref{cor:no-domain-bound}.

Fix any finite domain $E$ and any problem $\cP\mid\cR$ over $E$ that is
exactly equivalent to $\cF_k\mid\{B_k\}$ in the sense of
Definition~\ref{def:exact-labelled-equivalence}.  Suppose it has a valid common width-$r$
presentation $(\cP,\mathcal Q,N)$, so its induced right family is
$\cR=N\mathcal Q$ labelwise, with base
\begin{equation}
  N\in\C^{|E|\times2^r}.
  \label{eq:arbitrary-new-base}
\end{equation}
Let $P_i\in\cP$ be the unary left label corresponding to $\delta_i$.  Let
$B'_k\in\cR$ be the arity-$(k+2)$ right label corresponding to $B_k$.  By
validity of the new presentation, there is a Boolean block signature
$Q_B\in\MG_{(k+2)r}$ satisfying
\begin{equation}
  N^{\otimes(k+2)}Q_B=B'_k,
  \label{eq:new-right-preimage}
\end{equation}
and the transformed unary signatures
\begin{equation}
  g_i:=P_i N\in\MG_r
  \qquad(i\in D_3).
  \label{eq:new-unary-matchgates}
\end{equation}
Neither $N$ nor these domain signatures are assumed to have any particular
rank.  Contract the first two consecutive blocks of $Q_B$ with two copies
of $g_0$.  Exact matchgates are closed under ordered planar contraction, so
the result is a signature
\begin{equation}
  Q_0:=\Contr(Q_B;g_0,g_0,-,\ldots,-)\in\MG_{kr}.
  \label{eq:pinned-right-preimage}
\end{equation}
Here each dash denotes an uncontracted block, so $Q_0$ retains the final
$k$ blocks in their original order.

For each $x=(x_1,\ldots,x_k)\in\bits^k$, consider the closed planar star whose
centre is labelled $B_k$, whose two link leaves are labelled $\delta_0$,
and whose $j$th hard leaf is labelled $\delta_{2x_j}$.  Its original value
is $f_k^{\mathrm{int}}(x)$ by \eqref{eq:Bk-coordinate-table}.  Exact
labelled equivalence gives the same value after replacing these labels by
$B'_k,P_0,P_{2x_1},\ldots,P_{2x_k}$.  Equations
\eqref{eq:new-right-preimage}--\eqref{eq:pinned-right-preimage} and
associativity in the fixed block order give
\begin{align}
  f_k^{\mathrm{int}}(x)
  &=\left(P_0\otimes P_0\otimes
      \bigotimes_{j=1}^kP_{2x_j}\right)B'_k
  \notag\\
  &=\Contr(Q_B;g_0,g_0,g_{2x_1},\ldots,g_{2x_k})
  \notag\\
  &=\Contr(Q_0;g_{2x_1},\ldots,g_{2x_k}).
  \label{eq:star-extraction}
\end{align}

Set $\widetilde g_0:=g_0$ and $\widetilde g_1:=g_2$.  Then
\eqref{eq:star-extraction} is exactly a star-table representation with
sampler pair $(\widetilde g_0,\widetilde g_1)$.  Applying
\cref{prop:integer-obstruction-table} therefore gives
\begin{align}
  2^k
  &\le 2\delta(r)+\delta(kr)
  \notag\\
  &=2\left(1+\binom r2\right)
    +1+\binom{kr}{2}
  \notag\\
  &=3+r(r-1)+\frac{kr(kr-1)}{2}.
  \label{eq:exact-necessary-inequality}
\end{align}
This derivation uses no coordinate of the arbitrary domain $E$, no inverse
or rank hypothesis on $N$, and no relationship between $N$ and the
displayed base $M_k$.  It therefore applies to every presentation allowed
in Definition~\ref{def:base-width}, including $r=0$.

Finally, for every integer $r$ occurring in the minimization set of
Definition~\ref{def:base-width},
\eqref{eq:exact-necessary-inequality} gives the first inequality below; the
second inequality holds for every integer $r\ge0$:
\begin{equation}
  2^k-3
  \le r(r-1)+\frac{kr(kr-1)}2
  \le r^2\left(1+\frac{k^2}{2}\right).
  \label{eq:readable-inequality-derivation}
\end{equation}
Every such width $r$ must consequently satisfy
\begin{equation}
  r
  \ge
  \left\lceil
    \sqrt{\frac{2^k-3}{1+k^2/2}}
  \right\rceil.
  \label{eq:all-r-lower-bound}
\end{equation}
Taking the minimum over all equivalent presentations proves
\eqref{eq:main-exact-lower-bound} and
\eqref{eq:main-readable-lower-bound}.  Together with
\eqref{eq:muk-upper-bound} and
\cref{prop:valid-domain-three-presentation,prop:controlled-ranks,prop:two-center-coupling},
this proves parts~\textup{(i)}--\textup{(v)} of \cref{thm:main}.  The final
assertion that the presentation in part~\textup{(i)} lies in
alternative~\textup{(III)}
is established in the proof of \cref{thm:structural-trichotomy} in
\cref{sec:positive-boundary}.

\begin{corollary}[Explicit transcendental companion]
\label[corollary]{cor:explicit-transcendental-companion}
For the closed-form problem
$\cF_k\mid\{B_k^{\mathrm{tr}}\}$ of
\cref{prop:transcendental-companion-presentation}, every valid common
width-$r$ presentation whose induced problem is exactly equivalent to it,
even with arbitrary complex weights, satisfies
\begin{equation}
  2^k\le2\delta(r)+\delta(kr).
  \label{eq:transcendental-companion-exact-bound}
\end{equation}
Consequently,
\begin{equation}
  \mu_k^{\mathrm{tr}}
  \ge
  \left\lceil
    \sqrt{\frac{2^k-3}{1+k^2/2}}
  \right\rceil
  =\Omega\!\left(\frac{2^{k/2}}{k}\right).
  \label{eq:transcendental-companion-readable-bound}
\end{equation}
Thus the same unbounded-width result also has a simple closed-form witness,
at the cost of transcendental signature values.
\end{corollary}

\begin{proof}
Fix a candidate presentation $(\cP,\mathcal Q,N)$ from the corollary.  Let
$P_i$ be the unary labels corresponding to $\delta_i$, let $Q_B$ be the
Boolean preimage of the right label corresponding to
$B_k^{\mathrm{tr}}$, and put
\[
  g_i:=P_iN\in\MG_r,
  \qquad
  Q_0:=\Contr(Q_B;g_0,g_0,-,\ldots,-)\in\MG_{kr}.
\]
The same closed labelled stars used in \eqref{eq:star-extraction}, now with
$B_k^{\mathrm{tr}}$, give
\[
  f_k^{\mathrm{tr}}(x)
  =\Contr(Q_0;g_{2x_1},\ldots,g_{2x_k})
  \qquad(x\in\bits^k).
\]
With $\widetilde g_0:=g_0$ and $\widetilde g_1:=g_2$, this is a star-table
representation with sampler pair
$(\widetilde g_0,\widetilde g_1)$.
The explicit obstruction in
\cref{prop:transcendental-obstruction-table} proves
\eqref{eq:transcendental-companion-exact-bound}.  Substituting the
definition of $\delta$ and repeating
\eqref{eq:readable-inequality-derivation} proves the displayed lower bound
after taking the minimum over all equivalent presentations.
\end{proof}

\begin{proof}[Proof of \cref{cor:no-domain-bound}]
If such a function $h$ existed, applying it to the presentations from
\cref{prop:valid-domain-three-presentation} would give
$\mu_k\le h(3)$ for every $k$.  This contradicts $\mu_k\to\infty$.
\end{proof}

\section{The Positive Boundary: Support Geometry and Collapse}
\label{sec:positive-boundary}

\Cref{sec:arbitrary-equivalence} completes the quantitative negative
argument by showing that the constructed qutrit language cannot be
compressed to uniformly bounded width under arbitrary exact
re-presentation.  We now turn to the positive/classification lane of
\cref{fig:conceptual-roadmap}: under the qutrit, full-row-rank,
bi-essential, and port-deficient hypotheses, which realized support
geometries force a constant-rank common cover, and what structural form
remains when they do not?  The proof proceeds in the three stages indicated
in the roadmap.  Its first stage is support monotonicity, which propagates
primitive port deficiency through the all-left gadget closure and carries
the resulting supports faithfully through the full-row-rank base.

We use the tensor, gadget, ordered-matchgate, and support vocabulary
formalized in
\cref{sec:tensor-preliminaries,sec:holant-preliminaries,sec:matchgate-matrix-preliminaries,sec:support-terminology}.
Thus $V=\C^3$, $M$ is a width-$t$ qutrit base of row rank three, and
$\mathcal S_1,\mathcal S_2$ are the realized supports of
\eqref{eq:support-closures}.  These are supports of actual connected
all-left gadgets, not of their linear span.  Cyclic input rerooting and the
fixed reversal on the complementary output wires are those of
\cref{sec:support-terminology}, so ordinary
matrix multiplication below is ordered planar composition and no internal
wire permutation is implicit.  The parity endpoints $P^+,P^-$, monomial
flags, and connection-primality are those of
\cref{sec:matchgate-matrix-preliminaries,def:monomial-flag}.

The first observation records why primitive deficiency persists at a
dangling boundary edge.

\begin{lemma}[Boundary-support monotonicity]
\label[lemma]{lem:boundary-monotonicity}
Let $\Gamma$ be a planar gadget, and let a dangling edge $e$ be incident
with port $i$ of a primitive left tensor $F$.  If $T_\Gamma$ is the
boundary tensor of $\Gamma$, then
\[
  S_e(\Gamma)=S_{j(e)}(T_\Gamma)\subseteq S_i(F).
\]
\end{lemma}

\begin{proof}
Choose a minimal mode decomposition
\[
  F=\sum_{a=1}^{d}u_a\otimes F_a,
  \qquad
  S_i(F)=\Span\{u_1,\ldots,u_d\}.
\]
Contracting the rest of the gadget changes only the tensors multiplying
the $u_a$.  Thus
\[
  T_\Gamma=\sum_{a=1}^{d}u_a\otimes R_a
\]
for suitable residual boundary tensors $R_a$, proving the inclusion.
\end{proof}

Using the block-support convention fixed in
\cref{sec:support-terminology}, for a tensor $G$ on $q$ consecutive
$t$-bit blocks let
$S_i^{\rm blk}(G)\le\C^{\bits^t}$ be the image of the flattening that treats
the whole $i$th block as its row index and all other blocks as its column
index.  It differs from the support of one individual Boolean mode.

\begin{lemma}[Transformed block supports]
\label[lemma]{lem:transformed-block-support}
Suppose $q\ge1$, $T\in(V^*)^{\otimes q}$ is nonzero, and $M$ has row rank
three.  If $T$ has mode support $U=S_i(T)\le V^*$ at port $i$, then
\[
  S_i^{\rm blk}(TM^{\otimes q})=UM:=\{uM:u\in U\}.
\]
In particular, the two supports have the same dimension.
\end{lemma}

\begin{proof}
Write a minimal mode decomposition
\[
  T=\sum_{a=1}^{d}u_a\otimes T_a,
  \qquad U=\Span\{u_1,\ldots,u_d\},
\]
with the $T_a$ linearly independent.  Full row rank makes both
$u\mapsto uM$ and the corresponding tensor-power map injective.  Hence
\[
  TM^{\otimes q}
  =\sum_{a=1}^{d}(u_aM)\otimes
       (T_aM^{\otimes(q-1)})
\]
is again a minimal mode decomposition, with block support $UM$.
\end{proof}

The preceding support lemmas reduce the gadget closure to an incidence
problem among realized rays and planes, while full row rank carries these
supports without loss of dimension into exact matchgate block row spaces.
We next use pure-spinor geometry to ask when such Gaussian pieces must lie
inside one low-rank matchgate row space.  The following hull, decoder, and
common-annihilator lemmas turn the relevant incidence relations into
rank-four or rank-eight common-cover certificates.

\begin{lemma}[Rank-two Gaussian hull]
\label[lemma]{lem:gaussian-hull}
Let $R_1,R_2\le\C^{\bits^t}$ be row spaces of exact matchgate matrices.
Suppose
\[
  \dim R_1=\dim R_2=2,
  \qquad R_1\ne R_2,
  \qquad R_1\cap R_2\ne0.
\]
Then there is a two-input exact matchgate matrix $P$ of rank four such that
\[
  R_1+R_2\subseteq\row(P).
\]
\end{lemma}

\begin{proof}
Let parity refer to the $t$ ordered output wires.  Xia's rank-two
canonical form has one nonzero row-space dimension in each output parity,
so
\[
  R_i=R_i^+\oplus R_i^-,
  \qquad \dim R_i^+=\dim R_i^-=1.
\]
The intersection is parity invariant.  Since two distinct planes with
nonzero intersection meet in a line, choose a nonzero parity-homogeneous
$u\in R_1\cap R_2$.  In either original matchgate matrix an actual pinned
input row spans that parity component, so $u$ is, up to a nonzero scalar,
an exact matchgate vector.

Apply an invertible output matchgate transformation $C$ sending
$u$ to the vacuum $e_\varnothing$.  This is the vector case of Xia's
canonical form \cite[Theorem~4.1]{Xia2016}.  Put $R_i'=R_iC$ and choose a
matchgate matrix $A_i$ with $s_i$ input wires and row space $R_i'$.  After
input flips and a global rescaling, an actual input row is the pivot row
$e_\varnothing$.  In the pivot-Pfaffian chart of
\cref{sec:matchgate-matrix-preliminaries}, write its $s_i$ input and $t$
output Grassmann variables as $\xi$ and $\eta$.  Let $A_{{\rm in},i}$,
$B_i$, and $D_i$ denote,
respectively, the input--input, input--output, and output--output blocks of
the skew quadratic parameter matrix.  Up to the pivot scalar, the matchgate
matrix has generating function
\[
  \exp\!\left(
    \frac12\xi^{\trans}A_{{\rm in},i}\xi
    +\xi^{\trans}B_i\eta
    +\frac12\eta^{\trans}D_i\eta
  \right).
\]
The pivot row is exactly the output vacuum, so its two-output coordinates
(the coefficients of $\eta_p\eta_q$) vanish and $D_i=0$.  Put
\[
  q_i(\xi):=\frac12\xi^{\trans}A_{{\rm in},i}\xi.
\]
Exterior multiplication by $\exp(-q_i)$ is an invertible linear operator
on the input coefficient space, with inverse multiplication by
$\exp(q_i)$.  Applied coefficientwise in $\eta$, its ordinary coefficient
matrix performs invertible row operations and therefore preserves the row
space.  Since $q_i$ and $\xi^{\trans}B_i\eta$ are even exterior elements,
\[
  \exp(-q_i)\wedge
  \exp\bigl(q_i+\xi^{\trans}B_i\eta\bigr)
  =\exp(\xi^{\trans}B_i\eta).
\]
Thus no unproved boundary transformation is being used: the resulting
coefficient matrix is independently an exact matchgate matrix by the
principal-Pfaffian chart and has the same row space as $A_i$.  It is the
exterior-compound matrix of $B_i$, whose rank is $2^{\rank B_i}$.  Since
$\rank A_i=2$, the cross matrix has rank one, and therefore
\[
  R_i'=\Span\{e_\varnothing,a_i\}
\]
for a nonzero one-particle vector $a_i$.  The two $a_i$ are independent
because $R_1'\ne R_2'$.

Use $a_1,a_2$ as the rows of a $2$-by-$t$ cross matrix $B_0$.  The
coefficient matrix $P_0$ of $\exp(\xi^{\trans}B_0\eta)$ is an exact
two-input, $t$-output matchgate matrix by the Grassmann chart, and its
exterior-compound row space is
\[
  \row(P_0)=
  \Span\{e_\varnothing,a_1,a_2,a_1\wedge a_2\}
\]
and rank four.  It contains $R_1'+R_2'$.  Composing the ordered outputs
with $C^{-1}$ gives the exact matchgate matrix $P=P_0C^{-1}$ and
\[
  \row(P)\supseteq(R_1'+R_2')C^{-1}=R_1+R_2.
\]
\end{proof}

We next make explicit the exact-labelled interface between a Gaussian row
space and a smaller common base.

\begin{lemma}[Full-row-rank matchgate decoder]
\label[lemma]{lem:matchgate-decoder}
If $P$ is an $r$-input, $t$-output exact matchgate matrix of rank $2^r$,
then there is a $t$-input, $r$-output exact matchgate matrix $D$ such that
\[
  PD=I_{2^r}.
\]
\end{lemma}

\begin{proof}
Xia's ordered canonical form \cite[Theorem~4.1]{Xia2016} gives invertible
exact input and output matchgate transformations $A,C$ such that
\[
  APC=J_{r,t}=:J,
\]
where $J$ is the standard inclusion that copies the $r$ input modes to
$r$ active output modes and fixes the remaining $t-r$ output modes.  In
the ordered matrix convention fixed above, $JJ^{\trans}=I_{2^r}$, and the
standard projection $J^{\trans}$ is exact by
\cref{lem:ordered-clifford-realization}.  Concretely it is the reflected
nested-edge realization: its $r$ output ports inherit precisely the required
reversed physical order, the $r$ active wires are paired by the reflected
unit edges, and the other $t-r$ input wires carry the reflected zero-pins.
Therefore
\[
  D=CJ^{\trans}A
\]
is an exact ordered matchgate composition and
\[
  PD=A^{-1}JJ^{\trans}A=I_{2^r}.
\]
\end{proof}

\begin{lemma}[Exact-labelled cover compression]
\label[lemma]{lem:exact-cover-compression}
Let $M$ be a domain base with $t$ Boolean wires per block.  Suppose $P$ is
an exact matchgate matrix of rank $2^r$, with $r$ input wires and the same
$t$ ordered output wires, such that
\[
  \row(M)\subseteq\row(P).
\]
Then every valid presentation using $M$ has a presentation of the same
domain tensors, with the same labels, arities, and port orders, using a
width-$r$ base.
\end{lemma}

\begin{proof}
Write $M=NP$, with $N$ having $2^r$ columns, and let $D$ be the decoder in
\cref{lem:matchgate-decoder}.  For a right Boolean preimage $H$, put
$b:=\barity(H)$ and
$H'=P^{\otimes b}H$.  Matchgate closure under ordered planar composition
gives $H'\in\MG_{rb}$, and associativity gives, for the same right label,
\[
  M^{\otimes b}H=N^{\otimes b}H'.
\]
For a left label $F$, put $a:=\arity(F)$.  Its old lift is
\[
  FM^{\otimes a}=(FN^{\otimes a})P^{\otimes a}.
\]
Attach $D$ to every consecutive output block.  Exact matchgate closure and
$PD=I_{2^r}$ give
\[
  (FM^{\otimes a})D^{\otimes a}=FN^{\otimes a},
\]
so the new left lift is an exact matchgate.  Thus $N$ is the required
width-$r$ base.  This decoder form of Xia's cover-compression argument
\cite[Theorem~5.2]{Xia2016} is performed separately for every label.  No
label-dependent scalar, arity change, or port permutation occurs.
\end{proof}

\begin{lemma}[Isotropic-kernel cover]
\label[lemma]{lem:isotropic-kernel-cover}
Let $K\le W\oplus W^*$ be an isotropic subspace of dimension $t-r$, where
$0\le r\le t$.  Then there is an $r$-input, $t$-output exact matchgate
matrix $P_K$ of rank $2^r$ such that
\[
  \row(P_K)
  =\{u\in\Lambda^\bullet W: z\cdot u=0\text{ for every }z\in K\}.
\]
\end{lemma}

\begin{proof}
Write
\[
  \widehat L_K:=\{u\in\Lambda^\bullet W:z\cdot u=0
      \text{ for every }z\in K\}.
\]
By Witt extension, choose $g\in \mathrm O(W\oplus W^*)$ carrying $K$ to
\[
  K_0:=\Span\{f_{r+1}^*,\ldots,f_t^*\}\le W^*.
\]
Let $U_g$ be a Pin lift, acting on column spinors.  The intertwining identity
\eqref{eq:clifford-intertwining} shows that
\[
  U_g\widehat L_K=\widehat L_{K_0}.
\]
Put $C:=U_g^{\trans}$.  Both $C$ and
$C^{-1}=(U_g^{-1})^{\trans}$ are exact square matchgate transformations by
\cref{lem:ordered-clifford-realization}.  If $u$ is written as a coefficient
column, then
\[
  u^{\trans}C=(U_g u)^{\trans},
\]
so, in row coordinates,
$\widehat L_K C=\widehat L_{K_0}$.

The Clifford actions of $K_0$ are the standard annihilation operators, and
their joint kernel is exactly
$\Lambda^\bullet\Span(f_1,\ldots,f_r)$, of dimension $2^r$.  In the fixed
ordered coordinates this is $\row(J_{r,t})$.  Therefore
\[
  P_K:=J_{r,t}C^{-1}
\]
is an exact $r$-input, $t$-output matchgate matrix of rank $2^r$, and in
fact
\[
  \row(P_K)=\widehat L_{K_0}C^{-1}=\widehat L_K.
\]
This also realizes directly the common-annihilator dimension formula used
here; compare \cite[Theorem~1(1)]{Batista2014}.
\end{proof}

With support monotonicity and the Gaussian hull and compression machinery in
place, we return to the realized support geometry and perform the final
separation in \cref{fig:conceptual-roadmap}.  A ray-only closure forces
complete factorization; sufficient motion among planes or pure rays produces
a rank-four or rank-eight common cover; and the remaining plane-bearing
configuration is a one-plane monomial flag.  In the flag case incidence no
longer forces collapse, so we instead derive a uniform all-arity
Boolean-core-times-rays normal form.

\begin{theorem}[Connection-prime collapse]
\label{thm:connection-prime}
Let $(\cF,\cH,M)$ be a valid qutrit presentation, let
$\mathcal B=M\cH$ be its induced right family, and consider the all-left
gadget closure of $\cF\mid\mathcal B$.  Assume $\rank M=3$, $\mathcal B$ is
support-essential and every port of every domain tensor in
$\cF\cup\mathcal B$ is deficient.  (The displayed rank hypothesis is
redundant under right support-essentiality, but makes the domain of
\cref{def:monomial-flag} explicit.)  If this closure is connection-prime,
then $M$ has a common matchgate cover of rank at most eight.  Hence the
presentation collapses exactly and labelwise to width at most three.  If
two distinct planes occur in $\mathcal S_2$, the cover has rank four and
the width is at most two.
\end{theorem}

\begin{proof}
The stated rank condition also follows independently from right
essentiality, since all induced right mode supports lie in
$\operatorname{col}(M)$ and together span $\C^3$.  Now
\cref{lem:transformed-block-support} identifies every realized domain
support $U$ with the matchgate block support $UM$.  If distinct
$P,Q\in\mathcal S_2$ occur, then $PM$ and $QM$ are distinct rank-two
matchgate row spaces in the three-dimensional space $\row(M)$.  They meet
in a line.  The Gaussian-hull lemma puts their sum, hence $\row(M)$, in a
rank-four matchgate row space.  This is a common cover, and
\cref{lem:exact-cover-compression} gives width at most two.

Suppose instead that $\mathcal S_2=\{P\}$.  Connection-primality supplies a
realized line $L\not\subset P$, whose image $LM$ is a pure matchgate ray.
Since there is no monomial flag, there is another realized line $L'$
different from $P^+,P^-$ and $L$; its image is likewise a pure matchgate
ray.  A realized line inside $P$ can only be $P^+$ or $P^-$, because every
exact matchgate vector has one parity.  Thus $L'$ lies outside $P$.

If $t\le3$, the standard matrix $J_{t,t}=I_{2^t}$ has row space
$\C^{\bits^t}$ and rank $2^t\le8$.  It is a
common cover, so \cref{lem:exact-cover-compression} proves the conclusion.
Assume henceforth that $t\ge4$.

The space
\[
  \row(M)=PM+LM
\]
is parity invariant.  Let $\sigma\in\{+,-\}$ be the sign for which
$P^\sigma M$ has the same parity as $LM$, and write $-\sigma$ for the
opposite sign.  Choose nonzero generators
\[
  u\in P^\sigma M,\qquad v\in LM,\qquad w\in P^{-\sigma}M.
\]
The first parity component is the plane $\Span(u,v)$, while the other is
the line $\C w$.  The pure ray $L'M$ must have the first parity, since the
opposite component contains only $P^{-\sigma}M\subset PM$.  The choice of
$L'$ therefore implies
\[
  L'M=\C(au+bv),
  \qquad a,b\in\C^\times.
\]
Thus the projective line through the three distinct pure points
$[u],[v],[au+bv]$ contains three pure points.  The homogeneous vanishing
ideal of the spinor variety is
generated by quadrics \cite[Section~2.4]{Manivel2009}.  Every restricted defining
quadric has three zeros on this projective line and therefore vanishes
identically.  The whole line is a pure pencil, so the pure-pencil
classification of lines on the spinor variety
\cite[Example~4.12]{LandsbergManivel2003} supplies a fixed isotropic
$(t-2)$-space $K_{uv}$ contained in the annihilator of every point of the
line.  The two maximal isotropics $\Ann(u)$ and $\Ann(v)$ are distinct and
belong to the same spinor component.  In one component the codimension of
their intersection in either maximal isotropic is even
\cite[Section~2.2]{Manivel2009}.  Containment of
$K_{uv}$ makes that codimension at most two, while distinctness excludes
zero; it is therefore exactly two.  Hence
\[
  \dim(\Ann(u)\cap\Ann(v))=t-2.
\]

We justify the second annihilator dimension through an explicit Clifford
isometry.  Put $R:=PM$.  By Witt extension, choose
$g\in\mathrm O(W\oplus W^*)$ with $g\Ann(u)=W^*$, and let $U_g$ be a Pin
lift.  In row coordinates set $C:=U_g^{\trans}$.  The intertwining identity
\eqref{eq:clifford-intertwining} shows that $uC$ is a nonzero multiple of
the vacuum $e_\varnothing$, and
\cref{lem:ordered-clifford-realization} shows that $C$ and $C^{-1}$ are
exact matchgate transformations.  Hence $RC$ is a rank-two exact matchgate
row space containing the vacuum.  The pivot-chart argument in the proof of
\cref{lem:gaussian-hull} applies verbatim and gives
\[
  RC=\Span\{e_\varnothing,\alpha\}
\]
for a nonzero one-particle vector $\alpha\in W$.  A Pin lift either
preserves parity globally or swaps it globally.  Since $u$ and $w$ span the
two opposite-parity lines of $R$, the vector $wC$ spans $\C\alpha$.  Now
$\Ann(e_\varnothing)=W^*$, and
\[
  \Ann(e_\varnothing)\cap\Ann(\alpha)
  =\{\phi\in W^*: \phi(\alpha)=0\},
\]
a hyperplane of dimension $t-1$.  The same intertwining identity identifies
this intersection with
$g(\Ann(u)\cap\Ann(w))$.  Since $g$ is an isometry, it follows that
\[
  \dim(\Ann(u)\cap\Ann(w))=t-1.
\]
Both intersections lie in the $t$-dimensional maximal isotropic space
$\Ann(u)$.  Their intersection has dimension at least $t-3$ and
annihilates every vector in $\row(M)$.  Choose a
$(t-3)$-dimensional subspace
\[
  K\le\Ann(u)\cap\Ann(v)\cap\Ann(w).
\]
It is isotropic because $K\le\Ann(u)$.  Every vector in $\row(M)$ is
annihilated by $K$, so \cref{lem:isotropic-kernel-cover} with $r=3$ gives a
rank-eight exact matchgate matrix whose row space contains $\row(M)$.  This
is a common cover, and
\cref{lem:exact-cover-compression} gives width at most three
\cite[Theorem~5.2]{Xia2016}.
\end{proof}

The exceptional case has an all-arity normal form.  We first decode one
rank-two block.

\begin{lemma}[Rank-two block decoding]
\label[lemma]{lem:rank-two-decoding}
Let $R\le\C^{\bits^t}$ be the row space of a rank-two exact matchgate
matrix, and write
\[
  R=R^+\oplus R^-,
  \qquad \dim R^+=\dim R^-=1,
\]
where $R^+$ is even and $R^-$ is odd.  There are generators
$g_0\in R^+$, $g_1\in R^-$ and one invertible output matchgate transformation $C$,
depending only on $R$, such that
\[
  g_0C=e_{0^t}^{\trans},
  \qquad
  g_1C=e_{10^{t-1}}^{\trans}.
\]
Let $m\ge0$.  If an exact $mt$-wire matchgate tensor $G$ belongs to
$R^{\otimes m}$ and
\[
  G=\sum_{x\in\bits^m}h(x)
       g_{x_1}\otimes\cdots\otimes g_{x_m},
\]
then $h$ is an exact $m$-wire Boolean matchgate signature.
\end{lemma}

\begin{proof}
Xia's rank-two canonical form \cite[Theorem~4.1]{Xia2016}, in the terminology
of \cref{sec:matchgate-matrix-preliminaries}, gives an invertible output
matchgate transformation that sends $R$ to the row space of $J_{1,t}$,
after fixed output flips if needed.  Thus it sends $R$ to
\[
  \Span\{e_{0^t}^{\trans},e_{10^{t-1}}^{\trans}\}.
\]
The two parity lines map to the displayed coordinate lines.  If they are
interchanged, compose with a flip of the active output mode.  Rescale the
generators to obtain the stated equalities.

If $m=0$, the assertion is the convention $\MG_0=\C$.  Assume $m\ge1$.

Apply the same transformation to every consecutive block of $G$.
Matchgate signatures are closed under this blockwise composition, and the
result is
\[
  \sum_{x\in\bits^m}h(x)
       e_{x_10^{t-1}}^{\trans}\otimes\cdots\otimes
       e_{x_m0^{t-1}}^{\trans}.
\]
Pin the $t-1$ inactive wires of every block to zero, keeping the inherited
cyclic order on the active wires.  The remaining signature is exactly
$h$, so matchgate closure under pinning proves $h\in\MG_m$.
\end{proof}

\begin{lemma}[Rank-one stripping and flag normal form]
\label[lemma]{lem:flag-normal-form}
Assume $q\ge1$ and $M$ has row rank three, and let
$T\in(V^*)^{\otimes q}$ be nonzero with $TM^{\otimes q}\in\MG_{qt}$.

If every mode support of $T$ is a line, then $T$ is completely
decomposable.  More generally, suppose a plane $P\le V^*$ has the property
that $PM$ is the row space of a rank-two exact matchgate matrix, let
$P^+,P^-$ be its pullback parity endpoints, and suppose a line
$L\not\subset P$ satisfies
\[
  S_i(T)\in\{P,P^+,P^-,L\}
  \qquad\text{for every port }i.
\]
Choose the generators $g_0,g_1$ supplied by
\cref{lem:rank-two-decoding} for $R=PM$, and choose
$p_0\in P^+$ and $p_1\in P^-$ so that $p_aM=g_a$ for $a\in\bits$.  Put
$\iota(e_a)=p_a$ for $a\in\bits$.  If
\[
  I=\{i\in[q]:S_i(T)=P\},
\]
then there are nonzero $q_j$ spanning $S_j(T)$ for $j\notin I$ and an
exact Boolean matchgate $h_T\in\MG_{|I|}$ such that, in the original port
order,
\begin{equation}
  T=\sum_{\sigma\in\bits^I}h_T(\sigma)
       \bigotimes_{j=1}^{q}v_j(\sigma),
  \qquad
  v_j(\sigma)=
  \begin{cases}
    p_{\sigma(j)},&j\in I,\\
    q_j,&j\notin I.
  \end{cases}
  \label{eq:flag-normal-form}
\end{equation}
Here $h_T$ uses the cyclic order inherited by the ports in $I$.  If
$I=\varnothing$, it is a nullary scalar.  Necessarily $|I|=0$ or
$|I|\ge2$; in the latter case every mode of $h_T$ has rank two.  Each
$q_jM$, as well as $p_0M,p_1M$, is an exact pure matchgate vector.
\end{lemma}

\begin{proof}
If $S_i(T)$ is a line, minimality of the mode support gives
$T=q_i\otimes T'$ after placing port $i$ first only in the notation.
Tensoring by a nonzero vector does not change any other mode support, so
the operation may be repeated at every rank-one port.  If every port has
rank one, induction on $q$ gives complete decomposability.

For any stripped factor $q_j$, pin all other Boolean wires of the nonzero
lift $TM^{\otimes q}$ to a joint coordinate on which the residual factor
is nonzero.  The remaining block is a nonzero scalar multiple of $q_jM$,
so $q_jM$ is an exact matchgate vector.  The vectors $p_0M,p_1M$ are the
parity endpoints of a rank-two matchgate row space and are exact matchgate
vectors by canonical form.

In the second situation, stripping the ports outside $I$ leaves a core
$T_I\in P^{\otimes|I|}$.  For each $j\notin I$, choose a bit string $y_j$
with $(q_jM)(y_j)\ne0$.  Pinning the corresponding Boolean block of
$TM^{\otimes q}$ to $y_j$ leaves a nonzero scalar multiple of
$T_IM^{\otimes|I|}$.  Thus the core lift is an exact matchgate.  Expand
$T_I$ uniquely in the basis $(p_0,p_1)$ of $P$; its coefficient tensor is
$h_T$.  Applying \cref{lem:rank-two-decoding} with $R=PM$ proves
$h_T\in\MG_{|I|}$.  Finally, injectivity of $p\mapsto pM$ shows that a
rank-one mode of $h_T$ would make the corresponding mode support of $T$ a
line rather than $P$.  Hence every retained mode has rank two.
\end{proof}

We now prove the trichotomy stated in \cref{sec:main-results}.

\begin{proof}[Proof of \cref{thm:structural-trichotomy}]
The rank assumption is in fact also forced by right essentiality, because
every induced right mode support lies in $\operatorname{col}(M)$.  By
\cref{lem:boundary-monotonicity}, every nonzero boundary support in the
all-left closure has dimension one or two.  Since the primitive left
family occurs among the one-vertex gadgets, left essentiality gives
\begin{equation}
  \sum_{S\in\mathcal S_1\cup\mathcal S_2}S=V^*.
  \label{eq:closure-essential}
\end{equation}

If $\mathcal S_2=\varnothing$, every port of every nonzero boundary tensor
has rank one, so \cref{lem:flag-normal-form} gives alternative~\textup{(I)}.
Suppose $\mathcal S_2\ne\varnothing$.  If no monomial flag exists, the
closure is connection-prime by \eqref{eq:closure-essential}, and
\cref{thm:connection-prime} proves alternative~\textup{(II)}.  Otherwise
a monomial flag $(P,L)$ exists and gives the support inclusions in
alternative~\textup{(III)}.  Every positive-arity boundary tensor has an
exact lift by planar contraction of the transformed primitive matchgates, so
\cref{lem:flag-normal-form} applies with the same plane, rays, and encoding;
a nullary boundary tensor is already a scalar.
The alternatives are mutually exclusive and exhaustive according to
whether $\mathcal S_2$ is empty and, when nonempty, whether a monomial flag
exists.

For the final assertion, take the displayed presentation
$(\cF_k,\cH_k,M_k)$ from
\cref{prop:valid-domain-three-presentation}.  Every exposed edge of an
all-left gadget is incident with either a unary pin, whose support is one
of $\C e_0^*,\C e_1^*,\C e_2^*$, or a port of $X_{01}$, whose support is
\[
  P=\Span(e_0^*,e_1^*).
\]
Boundary-support monotonicity confines the resulting support to one of
those lines or to $P$.  A realized line inside $P$ maps under $M_k$ to a
pure matchgate ray inside the parity-split row space $PM_k$ and hence is
one of its two parity endpoints $\C e_0^*,\C e_1^*$.  Since $X_{01}$
realizes $P$ and $\delta_2$ realizes $L=\C e_2^*$, the displayed
presentation lies in alternative~\textup{(III)}.  Part~\textup{(v)} of
\cref{thm:main} makes $\mu_k$ tend to infinity, and
\cref{prop:controlled-ranks,prop:two-center-coupling} supply the remaining
operational properties.  This also proves the final flag-placement assertion
of \cref{thm:main}.
\end{proof}

The same geometry has a useful binary-context consequence.

\begin{corollary}[Binary-context alternative]
\label[corollary]{cor:binary-context-alternative}
Let $(\cF,\cH,M)$ be a valid qutrit presentation, let
$\mathcal B=M\cH$ be support-essential, assume $\rank M=3$, and consider the all-left gadget
closure of $\cF\mid\mathcal B$.  Suppose one realized rank-two plane $P$ and
one realized line $L\not\subset P$ span $V^*$.  As recorded in
\cref{sec:support-terminology}, $LM$ is then a pure matchgate ray.  Choose nonzero
generators $p_+\in P^+$, $p_-\in P^-$, and $\ell\in L$; the ordered triple
$(p_+,p_-,\ell)$ is an adapted ray basis.  A \emph{binary context} means the
boundary tensor of a connected all-left gadget with two ordered boundary
ports.  Then at least one of the following holds:
\begin{enumerate}[label=\textup{(\alph*)}]
  \item $M$ has a rank-four common matchgate cover;
  \item $M$ has a rank-eight common matchgate cover;
  \item every binary context is partial monomial
  in the adapted ray basis, meaning that its $3$-by-$3$ domain matrix has
  at most one nonzero entry in every row and column.
\end{enumerate}
Consequently, one non-partial-monomial binary context forces width at most
three.
\end{corollary}

\begin{proof}
Let $A$ be the domain matrix of a binary context, expressed in the chosen
adapted basis.  The stated full row rank also follows from
support-essentiality of $\mathcal B$, because every induced right mode
support is contained in $\operatorname{col}(M)$.  Its Boolean lift is an exact matchgate matrix, and this full
row rank preserves the matrix rank of $A$.  Since $\rank A\le3$ while a
nonzero exact matchgate matrix has power-of-two rank, the only possibilities
are ranks zero, one, and two.  Let $\sigma\in\{+,-\}$ be the sign for which
$P^\sigma M$ has the same parity as $LM$.

If $\rank A=1$, write $A=cxy^{\trans}$ with
$c\in\C^\times$ and nonzero coordinate columns $x,y\in\C^3$.  Full row
rank of $M$ makes both $xM$ and $yM$ nonzero.  Choose
$\beta\in\bits^t$ with $(yM)(\beta)\ne0$ and pin the second Boolean block of
the exact lift to the coordinate $\beta$.  The remaining block is the
nonzero scalar multiple $c(yM)(\beta)\,xM$, and hence $\C(xM)$ is a pure
matchgate ray.  Choosing $\alpha$ with $(xM)(\alpha)\ne0$ and pinning the
first block proves the same for $\C(yM)$.  Thus the local rays $\C x$ and
$\C y$ map to pure matchgate rays.
Call a ray \emph{noncoordinate} when it is not one of
$\C p_+,\C p_-,\C\ell$.  If either local ray is noncoordinate, parity places
it as a third pure
point on the same-parity line through $P^\sigma$ and $L$.  The pure-pencil
and common-annihilator argument in \cref{thm:connection-prime} gives a
rank-eight cover.  Otherwise $A$ is a scalar multiple of a matrix unit---a
matrix with one nonzero coordinate---and is partial monomial.

If $\rank A=2$, its row and column mode supports $S_1(A),S_2(A)$ are planes
that map to rank-two matchgate row spaces.  If either plane differs from $P$,
\cref{lem:gaussian-hull} gives a rank-four cover.  Otherwise both supports
equal $P$.  Expanding the lift in the opposite-parity endpoint rays
$P^+,P^-$, global matchgate parity allows only the two diagonal entries or
only the two antidiagonal entries.  Hence $A$ is partial monomial.  Rank
zero is trivial.
\end{proof}

\section{Discussion and Future Directions}
\label{sec:discussion}

The main results separate two sources of representation size.  Local domain
rank and port rank control the geometry visible on one edge, but they do not
control how much labelled information must be coordinated across many ports.
Support mobility eliminates that information by producing a small common
Gaussian cover.  A fixed plane-plus-ray flag leaves residual right-side data
on which the minimum exact width can grow.  The natural next step is to make
this distinction intrinsic to the language rather than to one displayed
presentation.

\subsection{From a support trichotomy to a language-level width theory}

The central open problem is a pointwise characterization of minimum exact
width.  The residual right tables defined from a flag presentation are useful
coordinates, but are not themselves invariants: another presentation may
change the flag, the domain, or the table.  One seeks a
presentation-independent invariant that yields a bounded-width realization
when small and quantitative lower bounds, valid for every exact presentation,
when large.  The labelled star-table obstruction is one such lower-bound
certificate for our family, but not yet a general characterization.

Two boundary cases sharpen the question.  Alternative~\textup{(I)} of
\cref{thm:structural-trichotomy} forces every nonzero positive-arity all-left
boundary tensor to factor, yet the minimum exact width of languages in this
branch is unclassified.  Is there a constant $C$ such that every language
admitting an alternative~\textup{(I)} qutrit presentation has minimum exact
width at most $C$, or is there such a family whose minimum widths tend to
infinity because of right-side data?  At the other boundary,
\cref{cor:binary-context-alternative} shows that one non-partial-monomial
binary context forces width at most three.  Is there a function $c(a,m)$ such
that, for qutrit presentations with primitive arity at most $a$ and at most
$m$ labels, failure of ray separability or flag confinement is always
witnessed by a planar context with at most $c(a,m)$ vertices and boundary
ports?  For rational or explicitly represented algebraic presentations, can
the remaining universal confinement conditions be decided from the primitive
signatures?

\subsection{Bounded arity and the source of unbounded information}

Our family uses a fixed number of labels and stores its growing table in the
increasing arity of $B_k$.  Does bounded arity alone force a bound $h(q,a)$ on
the minimum exact width of every exactly matchgate-realizable finite labelled
language on a $q$-element domain whose labels have arity at most $a$?  If not,
does a uniform bound $h(q,a,m)$ hold once the number of labels is also fixed?
A counterexample to this latter statement would require a storage mechanism
different from the growing-arity table used here.

For rational or integer languages, an effective version is equally
important.  With domain size $q$, arity bound $a$, label bound $m$, and
coefficient bit length $B$, can one give a useful quantitative or computable
bound $h(q,a,m,B)$?  Merely knowing that a finite maximum exists after all
four parameters are fixed would not explain the representation theory; the
goal is a structural and effective bound.

\subsection{Robustness under broader equivalence notions}

Exact labelled equivalence preserves the value of every ordered labelled
planar instance under the labelwise correspondence.  A first relaxation
would allow each label to
be replaced by a uniformly bounded exact planar gadget while retaining a
common base.  One can then ask whether the star-table lower bound survives,
perhaps after changing its sampling contexts.  More general gadget
interreducibility may alter arities and duplicate labels, so a width measure
must account for the size and uniformity of the substitutions as well as the
number of base wires.

Polynomial-time interreducibility is a still broader question and requires a
specified uniform input model for growing languages and coefficients.  The
present theorem makes no claim under either bounded-gadget equivalence or
polynomial-time interreducibility: it is a representation lower bound, not a
complexity lower bound.  Establishing which part of the obstruction is
invariant under bounded gadgets would be the most direct next test.

\subsection{Explicit and quantitative obstructions}

The main obstruction uses a positive-integer table chosen by Zariski
avoidance, but the proof supplies neither a simple formula nor a useful height
bound.  The transcendental companion has a closed formula and the same width
lower bound, at the cost of leaving the algebraic coefficient setting.  A
natural target is a uniformly generated positive-integer or rational family
with per-entry bit length polynomial in $k$ and construction time polynomial
in the total output length.

There is also a large quantitative gap.  The current construction gives
\[
  \Omega\!\left(\frac{2^{k/2}}{k}\right)
  \le \mu_k \le 2^k+3.
\]
It is unknown whether the dimension obstruction can be strengthened, whether
the block-Pfaffian realization can be compressed, or whether the true width
lies at an intermediate scale.  Resolving this gap would distinguish the
intrinsic representational cost of the table from the overhead of the present
universal interpolation.

\subsection{Beyond qutrits and intrinsic compression}

The qutrit case is special because every deficient nonzero mode support is a
ray or a plane.  Larger domains admit intermediate-dimensional supports and
richer incidence patterns among Gaussian row spaces.  A higher-domain theory
would likely organize these supports through Grassmannian or spinorial
incidence data and determine which kinds of mobility force a
domain-controlled common cover.  For each fixed domain size, does this
incidence geometry admit a finite stratification analogous to
\cref{thm:structural-trichotomy}?  If such stratifications exist, do their
types eventually stabilize as the domain grows, or do genuinely new
exceptional strata continue to appear?

From the free-fermionic viewpoint, the same question asks when a common exact
Gaussian tensor representation admits simultaneous mode compression.  This
is an algebraic reformulation, not a claim about physical qubit compression,
but it may suggest invariants that are invisible in the original base
coordinates.

Thus the trichotomy identifies support mobility as the local mechanism that
forces compression, while the construction shows that a fixed flag can
preserve labelled information of arbitrarily large exact cost.  The remaining
challenge is to recognize that information intrinsically---from the language
rather than from a chosen presentation.

\noindent\textbf{AI Disclosure:}
We used OpenAI Codex to assist with developing and checking mathematical
arguments, reviewing the literature, revising the exposition and organization,
and preparing the \LaTeX{} manuscript.  The tool materially affected
Sections~1--11.  More details on the affected proof architecture and technical
components can be found in Section~1.3.  The authors verified the correctness
and originality of all content, including references.

\end{document}